\documentclass{SciPost}
\usepackage{hyperref}
\hypersetup{
    colorlinks=true,
    linkcolor={red!50!black},
    citecolor={blue!50!black},
    urlcolor={blue!80!black},
    pdftitle={AlgebrasTQD: Computing Algebras and Interfaces in Twisted Quantum Doubles},
    pdfauthor={Alison Warman}
}
\usepackage[bitstream-charter]{mathdesign}
\SetMathAlphabet{\mathsf}{bold}{OT1}{cmss}{bx}{n} 
\usepackage{booktabs}
\usepackage{tabularx}
\usepackage{makecell}
\usepackage{float}
\usepackage{microtype}
\usepackage{listings}
\usepackage{color}
\usepackage[dvipsnames]{xcolor}
\usepackage{bm}
\usepackage{tikz}
\usetikzlibrary{positioning}
\usepackage{amsthm}

\newtheoremstyle{plainsp}{0.75\baselineskip}{0pt}{\itshape}{}{\bfseries}{.}{ }{}
\makeatletter
\renewenvironment{proof}[1][\proofname]{\par\pushQED{\qed}\normalfont\topsep\z@\trivlist\item[\hskip\labelsep\itshape#1\@addpunct{.}]\ignorespaces}{\popQED\endtrivlist\@endpefalse}
\makeatother
\theoremstyle{plainsp}

\newtheorem{theorem}{Theorem}[section]
\newtheorem{lemma}[theorem]{Lemma}

\theoremstyle{definition}

\DeclareSymbolFont{usualmathcal}{OMS}{cmsy}{m}{n}
\DeclareSymbolFontAlphabet{\mathcal}{usualmathcal}
\newcommand{\be}{\begin{equation}}
\newcommand{\ee}{\end{equation}}
\newcommand{\ba}{\begin{aligned}}
\newcommand{\ea}{\end{aligned}}
\newcommand{\cZ}{\mathcal{Z}}
\newcommand{\cL}{\mathcal{L}}
\newcommand{\cA}{\mathcal{A}}
\newcommand{\cS}{\mathcal{S}}
\newcommand{\cI}{\mathcal{I}}
\newcommand{\cP}{\mathcal{P}}
\newcommand{\cC}{\mathcal{C}}
\renewcommand{\Vec}{\mathsf{Vec}}
\newcommand{\Rep}{\mathsf{Rep}}
\newcommand{\Z}{\mathbb{Z}}

\newcommand{\eps}{\epsilon}

\newcommand{\Tr}{\operatorname{Tr}}
\newcommand{\diag}{\operatorname{diag}}
\newcommand{\folded}{\mathrm{folded}}
\newcommand{\ol}{\overline}

\newcommand{\sym}{\mathrm{sym}}
\newcommand{\phys}{\mathrm{phys}}
\newcommand{\ket}[1]{{|#1\rangle}}
\usepackage[skip=0.75\baselineskip plus 2pt]{parskip}

\newcommand{\paragraf}[1]{\par{\bf #1}}

  \usepackage{enumitem}
  \setlist[itemize]{itemsep=0pt,parsep=0pt,labelindent=1em,leftmargin=*}
  \setlist[itemize,1]{itemsep=0.3\baselineskip}
  \setlist[enumerate]{itemsep=0pt,parsep=0pt,labelindent=1em,leftmargin=*}
  \setlist[enumerate,1]{itemsep=0.3\baselineskip}

\newcommand{\om}{\omega}
\newcommand{\III}{\text{III}}
\newcommand{\cO}{\mathcal{O}}
\newcommand{\Irr}{\text{Irr}}
\newcommand{\id}{\text{id}}

\renewcommand{\dim}{\operatorname{dim}}
\makeatletter
\providecommand\xlabel[2][]{\phantomsection\def\@currentlabelname{#1}\label{#2}}
\makeatother

\usepackage[scaled=0.82]{DejaVuSansMono}
  \makeatletter
  \begingroup\fontencoding{T1}\fontfamily{DejaVuSansMono-TLF}\selectfont\global\hyphenchar\font=\m@ne\endgroup
  \DeclareFontFamily{T1}{DejaVuSansMono-TLF}{\hyphenchar\font=\m@ne}
  \makeatother
\usepackage[skins,breakable,listings]{tcolorbox}
\definecolor{codeframe}{HTML}{D4D4D4}
\definecolor{coderow}{HTML}{F1F1F1}
\definecolor{codekw}{HTML}{1565C0}
\definecolor{codestr}{HTML}{2E7D32}
\definecolor{codecomment}{HTML}{6E7781}
\newcommand{\codezebra}[1]{\ifodd\value{lstnumber}\rlap{\hspace{-7pt}\color{coderow}\rule[-0.75ex]{\dimexpr\linewidth+14pt\relax}{\dimexpr\baselineskip-1pt\relax}}\fi}
\lstdefinestyle{gapsnippet}{
    language=GAP,
    basicstyle=\ttfamily\small,
    keywordstyle=\color{codekw},
    stringstyle=\color{codestr},
    commentstyle=\color{codecomment},
    showstringspaces=false,
    columns=fullflexible,
    keepspaces=true,
    breaklines=true,
    breakindent=1em,
    postbreak={},
    numbers=left,
    numbersep=0pt,
    numberstyle=\codezebra,
    xleftmargin=7pt,
    xrightmargin=7pt,
    lineskip=1pt,
    aboveskip=0pt,
    belowskip=0pt,
}
\tcbset{
    gapbox/.style={
        enhanced,
        breakable,
        listing only,
        listing options={style=gapsnippet},
        colback=white,
        colframe=codeframe,
        boxrule=0.6pt,
        arc=3pt,
        left=0pt,
        right=0pt,
        top=3pt,
        bottom=3pt,
        boxsep=0pt,
        before skip=0.8\baselineskip,
        after skip=0.8\baselineskip,
        finish unbroken and first={\node[anchor=north east,inner sep=0pt] at ([xshift=-5pt,yshift=-4pt]frame.north east) {\includegraphics[height=13pt]{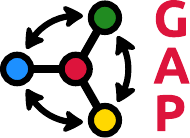}};},
    },
}
\DeclareTCBListing{gapcode}{ !O{} }{gapbox,#1}

\lstdefinestyle{gapoutput}{
    style=gapsnippet,
    language={},
    basicstyle=\ttfamily\footnotesize,
}
\tcbset{
    outbox/.style={
        enhanced,
        breakable,
        listing only,
        listing options={style=gapoutput},
        colback=gray!6,
        colframe=codeframe,
        boxrule=0.6pt,
        arc=3pt,
        left=0pt,
        right=0pt,
        top=3pt,
        bottom=3pt,
        boxsep=0pt,
        before skip=0.8\baselineskip,
        after skip=0.8\baselineskip,
    },
}
\newtcblisting{gapout}{outbox}

\tcbset{
    codebox/.style={
        enhanced,
        breakable,
        listing only,
        listing options={style=gapsnippet},
        colback=white,
        colframe=codeframe,
        boxrule=0.6pt,
        arc=3pt,
        left=0pt,
        right=0pt,
        top=3pt,
        bottom=3pt,
        boxsep=0pt,
        before skip=0.8\baselineskip,
        after skip=0.8\baselineskip,
    },
}
\newtcblisting{code}[1][]{codebox,#1}

\fancypagestyle{SPstyle}{
\fancyhf{}

\fancyfoot[C]{\textbf{\thepage}}
}

 \makeatletter  \long\def\@makefntext#1{\def\baselinestretch{1}\parindent\z@\noindent\@makefnmark\ #1\par}
  \makeatother

\begin{document}
\microtypesetup{protrusion=false,expansion=false}
\pagestyle{SPstyle}

\begin{center}{\Large\textbf{\color{scipostdeepblue}{AlgebrasTQD: Computing Algebras and Interfaces\\[1mm] in Twisted Quantum Doubles}}}\end{center}
\begin{center}\textbf{Alison Warman}\end{center}
\begin{center}{\itshape Mathematical Institute, University of Oxford,\\Andrew Wiles Building, Woodstock Road, Oxford OX2 6GG, United Kingdom}
\end{center}

\section*{\color{scipostdeepblue}{Abstract}}
\textbf{\boldmath{We present AlgebrasTQD, an open-source code written in GAP (Groups, Algorithms, and Programming), for computing condensable algebras and the corresponding topological boundaries and interfaces of twisted quantum doubles $D^\omega(G)$. These algebras can be used to classify 1+1d gapped and gapless phases whose symmetry is related by gauging to a finite group $G$ with anomaly $\omega\in H^3(G,U(1))$, to study 2+1d anyon proliferation and to construct logical gates in non-Abelian topological quantum codes. The functions in AlgebrasTQD compute all condensable algebras $A(H,F,\gamma,\epsilon)$ with their algebra structure and anyons, sort them into isomorphism classes and identify the twin algebras that are inequivalent but share the same anyon decomposition. AlgebrasTQD also determines all the subalgebra relations, the corresponding Hasse diagram and the types of 1+1d phases for any chosen, possibly non-invertible, group-theoretical symmetry. For non-maximal algebras, the code computes the reduced topological order with its $3$-cocycle and the map of anyons across the topological interface. Finally, AlgebrasTQD includes functions that construct certain non-Clifford logical gates via non-Abelian surface codes. We summarize the formulas that are implemented in the AlgebrasTQD code, and we illustrate its use with Abelian and non-Abelian groups including $\mathbb{Z}_2^3$ with $\omega_\text{III}$, $D_8$, $A_4$ with and without anomaly, twin algebras and twin phases, and the Clifford-hierarchy gates in quantum doubles of dihedral groups. All arithmetic is exact and AlgebrasTQD can be used with any finite group $G$ and 3-cocycle $\omega$, limited only by the time and size of the calculation. The code \cite{AlgebrasTQD} is available at \url{https://github.com/AlisonWarman/AlgebrasTQD}.}}

\vspace{\baselineskip}

\vspace{10pt}
\noindent\rule{\textwidth}{1pt}
\tableofcontents
\noindent\rule{\textwidth}{1pt}
\vspace{10pt}

\section{Introduction}
\label{sec:intro}

Finite group gauge theories are among the simplest and best understood topological phases of matter. In 2+1d, the gauge theory of a finite group $G$ with Dijkgraaf-Witten twist $\omega\in H^3(G,U(1))$~\cite{Dijkgraaf:1989pz,deWildPropitius:1995cf} is described by the twisted quantum double $D^\omega(G)$~\cite{Roche:1990hs}, and is realized on the lattice by the model of Kitaev~\cite{Kitaev:1997wr} and its twisted generalization~\cite{Hu:2012wx}. 

Twisted quantum doubles appear in a multitude of different settings across physics. In the context of generalized symmetries, $D^\omega(G)$ is the Symmetry Topological Field Theory (SymTFT) of 1+1d systems with symmetry $G$ and anomaly $\omega$, or any symmetry related to it by gauging, and it organizes their gapped and gapless phases and the transitions between them~\cite{Gaiotto:2020iye,Apruzzi:2021nmk,Freed:2022qnc,Kong:2020cie,Bhardwaj:2024qrf,Schafer-Nameki:2023jdn}. As a 2+1d topological order, it admits gapped boundaries and topological interfaces, and phase transitions driven by anyon proliferation~\cite{Cheng:2026qax,Cheng:2026hgu,Cheng:2026dxn,Schafer-Nameki:2026hgx,Gai:2026vgz,Eck:2026chq}. In quantum computing, the toric code $D(\Z_2)$ is the prototypical topological quantum memory, while quantum doubles of non-Abelian groups provide non-Clifford logical gates and routes to universal fault-tolerant quantum computation, e.g.~\cite{Kitaev:1997wr,Cui:2019lvb,Laubscher:2019rss,Davydova:2025ylx,Huang:2025ump,Huang:2025cvt,Kobayashi:2025cfh,Warman:2025hov,Manjunath:2026swm,Ellison:2021vth,Li:2026wly,Ding:2026bhp}. 

In all these applications, the central tools are the condensable algebras of $D^\omega(G)$, which characterize the ways that anyons can end consistently. They are classified by group-theoretical data~\cite{Davydov2009ModularIF,davydov2017lagrangian}, but finding them by hand quickly becomes impractical. For example, the quantum double of the dihedral group of order $8$ already has 38 different algebras \cite{Bhardwaj:2024qrf} while for groups of order $32$ it has hundreds. Moreover, an algebra is not determined by its anyons: twin algebras with the same anyons define different phases~\cite{Gai:2026hjk,Warman:2026gfz} and the full algebra data involves computations of projective representations and group cohomology.

AlgebrasTQD \cite{AlgebrasTQD} is an open-source GAP~\cite{GAP4} code for computations in the twisted quantum double $D^\omega(G)$ of any finite group $G$ and any $3$-cocycle $\omega$, with exact arithmetic. It provides the group cocycles that describe anomalies and SPT phases and sets up each theory with its modular data, fusion rules and readable anyon labels. Its core is the computation of all the condensable algebras: the code first computes the group-theoretical data that defines them and then decomposes the algebras into anyons. It then sorts them into isomorphism classes and computes the subalgebra relations between them. It also determines the corresponding 1+1d phases. Furthermore, every non-maximal algebra defines an interface to a reduced topological order, and the code computes the map of anyons across it. Finally, it constructs certain non-Clifford logical gates using non-Abelian surface codes. The AlgebrasTQD code therefore substantially facilitates the computation of gapped and intrinsically gapless phases~\cite{Bhardwaj:2024qrf,Warman:2024lir}, topological mixed states~\cite{Schafer-Nameki:2025fiy}, the determination of twin algebras and twin phases~\cite{Gai:2026hjk,Warman:2026gfz} and the construction of certain logical gates via non-Abelian codes~\cite{Huang:2025ump,Warman:2025hov}.

This paper is organized as follows. Section~\ref{sec:motivation} reviews the physical and mathematical settings and motivations for the code, Section~\ref{sec:theory} summarizes the theory and the formulas implemented in AlgebrasTQD, with proofs of code optimizations in Appendix~\ref{app:proofs}. Section~\ref{sec:install} describes how to run the code and its structure, while Section~\ref{sec:examples} illustrates the example scripts. Conclusions and future directions are discussed in Section~\ref{sec:conclusion}.

\section{Physical Settings and Motivations}
\label{sec:motivation}
\subsection{Phases and Transitions in 1+1d}
\label{sec:mot1d}
In this Section, we summarize the physical and mathematical settings which motivated the AlgebrasTQD code.

\paragraf{The Categorical Landau Paradigm.} In the Landau paradigm~\cite{Landau:1937obd} the gapped phases of a system with a finite symmetry group $G$ are labeled by the subgroup $H\leq G$ that is preserved, and a continuous transition between two phases requires one unbroken subgroup to be contained in the other. In 1+1d the unbroken symmetry can in addition be realized as a symmetry-protected topological (SPT) phase, so that the gapped phases are pairs $(H,\gamma)$ with $\gamma\in H^2(H,U(1))$~\cite{Chen:2010zpc,Fidkowski:2010jmn,Schuch:2011niz,Pollmann:2009ryx,Chen:2011pg}, distinguished by string order parameters~\cite{Perez-Garcia:2008hhq,Pollmann:2012,Haegeman:2012pex}. The $\Z_2\times\Z_2$ SPT (Haldane phase) is the classic example in which an SPT transition is actually a symmetry-breaking transition in disguise, revealed by a non-local transformation~\cite{Kennedy:1992ifl,kennedy1992hidden}. Generalized symmetries~\cite{Gaiotto:2014kfa,Schafer-Nameki:2023jdn,Shao:2023gho,Bhardwaj:2023kri}, which can be non-invertible (such as the $\Rep(G)$ symmetry obtained by gauging $G$), extend this framework to a categorical Landau paradigm~\cite{Bhardwaj:2023fca,Chen:2025uno,Bhardwaj:2024qrf,Moradi:2022lqp,Chatterjee:2022tyg}.

\paragraf{Symmetry Topological Field Theory.} The tool that enables systematic classifications of phases with global symmetries is the Symmetry Topological Field Theory (SymTFT)~\cite{Ji:2019jhk,Kong:2020cie,Chatterjee:2022kxb,Chatterjee:2022jll,Gaiotto:2020iye,Apruzzi:2021nmk,Freed:2022qnc,Kaidi:2022cpf,Moradi:2022lqp,Inamura:2023ldn}, Figure~\ref{fig:symtft}. The symmetry $\cS$ of a 1+1d system is a gapped ``symmetry'' boundary of a 2+1d topological bulk (whose anyons form the Drinfeld center $\cZ(\cS)$) and the dynamics occur on the other ``physical'' boundary. The system is recovered by shrinking the topological SymTFT slab. For $G$ with anomaly $\omega$, and for every symmetry related to it by gauging a subgroup with discrete torsion, the bulk is the twisted quantum double $D^\omega(G)$, whose anyons form $\cZ(\Vec_G^\omega)$. On the lattice, models with these related symmetries are mapped into each other by matrix product operator intertwiners and low-depth circuits~\cite{Lootens:2021tet,Lootens:2022avn,Lootens:2023wnl}. The generalized charges of $\cS$ are the topological defects in $\cZ(\cS)$~\cite{Bhardwaj:2023ayw}: for a 2+1d bulk, an anyon that ends on both boundaries gives rise to a local order parameter, while one that ends only on the physical boundary produces a string order parameter.

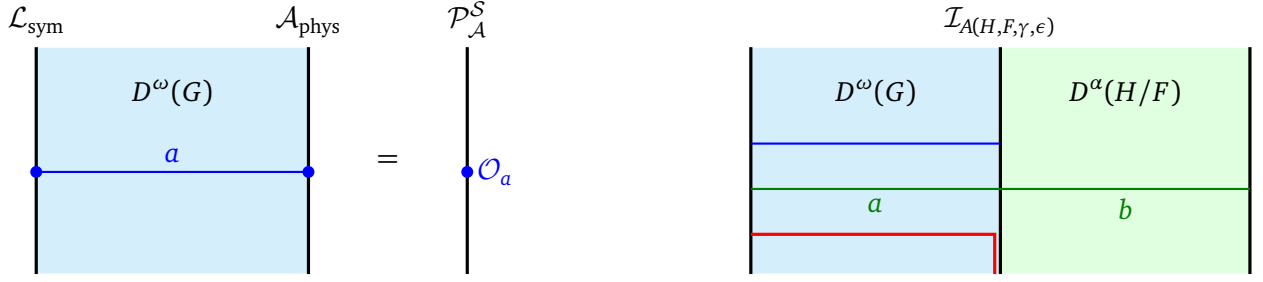
\begin{figure}
\centering
\begin{tikzpicture}[scale=1.5]
\begin{scope}[shift={(0,0)}]
    \draw[white,fill=cyan!15] (0,0) rectangle (2.4,2);
    \draw[very thick] (0,0)--(0,2) (2.4,0)--(2.4,2);
    \node at (1.2,1.6) {$D^\omega(G)$};
    \node[above] at (0,2) {$\cL_\sym$};
    \node[above] at (2.4,2) {$\cA_\phys$};
    \draw[thick,blue] (0,0.9)--(2.4,0.9);
    \fill[blue] (0,0.9) circle (1.5pt) (2.4,0.9) circle (1.5pt);
    \node[above,blue] at (1.2,0.9) {$a$};
    \node at (3.1,1) {$=$};
    \draw[very thick] (3.8,0)--(3.8,2);
    \node[above] at (3.8,2) {$\cP^{\cS}_{\cA}$};
    \fill[blue] (3.8,0.9) circle (1.5pt);
    \node[right,blue] at (3.8,0.9) {$\mathcal{O}_a$};
\end{scope}
\hfill
\begin{scope}[shift={(6.3,0)}]
    \draw[white,fill=cyan!15] (0,0) rectangle (2.2,2);
    \draw[white,fill=green!12] (2.2,0) rectangle (4.4,2);
    \draw[very thick] (0,0)--(0,2) (2.2,0)--(2.2,2) (4.4,0)--(4.4,2);
    \node at (1.1,1.6) {$D^\omega(G)$};
    \node at (3.3,1.6) {$D^\alpha(H/F)$};
    \node[above] at (2.2,2) {$\cI_{A(H,F,\gamma,\eps)}$};
    \draw[thick,blue] (0,1.15)--(2.2,1.15);
    \draw[thick,green!50!black] (0,0.75)--(4.4,0.75);
    \node[below,green!50!black] at (1.1,0.75) {$a$};
    \node[below,green!50!black] at (3.3,0.75) {$b$};
    \draw[very thick,red] (0,0.35)--(2.15,0.35)--(2.15,0);
\end{scope}
\end{tikzpicture}
\caption{Left: the SymTFT slab. The symmetry boundary $\cL_\sym$ fixes the symmetry $\cS$, the physical boundary is specified by a condensable algebra $\cA_\phys$, and anyons $a$ that end on both boundaries become the order parameters $\mathcal{O}_a$ of the phase $\cP^\cS_\cA$. Right: a non-Lagrangian algebra defines an interface to the reduced topological order. Anyons are condensed on the interface (blue), confined (red), or pass through as anyons $b$ of the reduced topological order (green): this is encoded in the map of anyons $b\mapsto a$ of \eqref{eq:map}.}
\label{fig:symtft}
\end{figure}

\paragraf{Phases and Transitions.} After fixing a symmetry, its phases are classified by condensable algebras $\cA$ on the physical boundary. When $\cA$ is Lagrangian the phase is gapped~\cite{Thorngren:2019iar,Bhardwaj:2023fca,Bhardwaj:2023idu}, and when it is not the phase is gapless, with the charges that are neither condensed nor confined carried by gapless degrees of freedom~\cite{Chatterjee:2022tyg,Wen:2023otf,Bhardwaj:2023bbf,Huang:2023pyk,Bhardwaj:2024qrf,Bottini:2025hri}. The number of gapped ground states, or of gapless universes, is the scalar product $n=\sum_an_a(\cL_\sym)\,n_a(\cA)$ of the vectors of anyon multiplicities of the symmetry and physical algebras, which distinguishes symmetry-preserving ($n=1$) from spontaneous symmetry breaking (SSB, $n>1$) phases. The inclusions between condensable algebras, encoded in a Hasse diagram~\cite{Bhardwaj:2024qrf}, determine the possible symmetry-preserving deformations. A gapless phase is intrinsically gapless when it cannot be deformed into a gapped phase with the same symmetry-breaking pattern~\cite{Wen:2023otf,Bhardwaj:2024qrf,Warman:2024lir}. The interface of a non-Lagrangian algebra to the reduced topological order, Figure~\ref{fig:symtft}, and the map of anyons across it, can be used to construct new transitions from a simpler ``input'' transition for the reduced symmetry~\cite{Chatterjee:2022tyg,Bhardwaj:2023bbf,Bhardwaj:2024qrf}. On the lattice, these phases and transitions can be realized by anyon chains or simpler spin chains~\cite{Aasen:2016dop,Aasen:2020jwb,Bhardwaj:2024kvy,Warman:2024lir,Inamura:2025cum,Lootens:2024gfp}. The computations of all condensable algebras and topological interfaces in $D^\omega(G)$ can be carried out with the AlgebrasTQD code described in this paper. 

\paragraf{Twin Phases.} It was conjectured that every 1+1d transition is related by gauging to a symmetry-breaking transition~\cite{Chen:2025uno}, as in the examples in~\cite{Kennedy:1992ifl,kennedy1992hidden,Zhang:2022wwn,Chen:2025ivo,Bhardwaj:2024qrf,Bhardwaj:2024wlr,Bhardwaj:2024kvy,Chatterjee:2024ych,Warman:2024lir,Bottini:2025hri,Chen:2026fai,Ravindran:2026jjn}. However, it was recently discovered that twin algebras~\cite{Gai:2026hjk} provide counterexamples~\cite{Warman:2026gfz}. Two Lagrangian algebras with the same anyons define two gapped phases whose order parameters lie in the same multiplet of generalized charges, and since the anyons ending on the symmetry boundary are the same for each twin, the two phases have the same number of ground states for every gauging of the symmetry: transitions between them are intrinsic deconfined quantum critical points~\cite{Senthil:2003eed,Senthil:2023vqd} in 1+1d~\cite{Warman:2026gfz}. What distinguishes twin phases is the multiplication of their order parameters~\cite{Gai:2026hjk}, or generalized string order parameters~\cite{Pollmann:2012,Haegeman:2012pex,Kobayashi:2025pxs}: both are captured by the full algebra data $A(H,F,\gamma,\eps)$ which our AlgebrasTQD code computes.

\subsection{Twisted Quantum Doubles in 2+1d}
\label{sec:mot2d}

The second physical setting of relevance for the AlgebrasTQD code is the 2+1d topological order $D^\omega(G)$ itself. It contains the simplest non-Abelian topological orders with an explicit lattice realization~\cite{Kitaev:1997wr,Hu:2012wx,Levin:2004mi,Bultinck:2015bot}, and certain examples have been realized on quantum processors~\cite{Iqbal:2023wvm}. A condensable algebra is a set of anyons that can condense consistently~\cite{Bais:2008ni,Kong:2013aya}: Lagrangian algebras describe gapped boundaries, while non-Lagrangian ones determine interfaces to smaller topological orders together with the maps of anyons across them~\cite{Beigi:2010htr,Lan:2014uaa,Kong:2017hcw,Cong:2017ffh,Lootens:2020mso,Williamson:2017uzx}. The distinction between condensation and proliferation transitions has been extensively studied recently for Abelian and non-Abelian anyons~\cite{Cheng:2026qax,Cheng:2026hgu,Cheng:2026dxn,Schafer-Nameki:2026hgx,Eck:2026chq}. The condensable algebras in $D^\om(G)$, the reduced topological orders they determine and the maps of anyons across interfaces (shown on the right of Figure~\ref{fig:symtft}), can be computed with AlgebrasTQD.

\subsection{Quantum Computing with Non-Abelian Codes}
\label{sec:motqc}

Kitaev's quantum double $D(G)$ on a patch with boundaries defines a non-Abelian surface code (with possible 3-cocycle twist $\omega$)~\cite{Kitaev:1997wr,Beigi:2010htr,Davydova:2025ylx,Huang:2025cvt,Huang:2025ump,Kobayashi:2025cfh,Warman:2025hov,Manjunath:2026swm,Ellison:2021vth,Fechisin:2026ywx,Warman:2026flz} with boundaries determined by Lagrangian algebras. Hybrid lattice surgery between Abelian and non-Abelian patches produces non-Clifford gates and magic states that can be topologically teleported into the standard surface code~\cite{Huang:2025ump,Ding:2026bhp}, while stacking a group $2$-cocycle with boundary counterterms on a triangular patch gives constant-depth logical gates at every level of the Clifford hierarchy in two dimensions~\cite{Warman:2025hov}. AlgebrasTQD computes the data for these protocols and returns the logical gate with its Clifford-hierarchy level. 

\subsection{Mathematics}
\label{sec:motmath}

Another application of the AlgebrasTQD code is to the mathematics of connected \'etale algebras~\cite{Davydov:2010kfz,davydov2017lagrangian} (called condensable algebras in the physics literature~\cite{Bais:2008ni,Kong:2013aya,Gai:2026hjk}). AlgebrasTQD can compute, for every finite group and cocycle class, the complete list of algebras up to isomorphism, their decomposition into simple objects, the braided autoequivalences of $\cZ(\Vec_G^\omega)$~\cite{Davydov2009ModularIF,Davydov:2010kfz} and the module categories that the Lagrangian algebras define~\cite{Ostrikmodule,Natale2017}.

\section{Condensable Algebras and Code Functions}
\label{sec:theory}

In this section, we list the formulas implemented by the main functions of AlgebrasTQD, based on the mathematical results of~\cite{davydov2017lagrangian,Gai:2026hjk}. The optimizations that the code adds are proved in App.~\ref{app:proofs}. This section serves as a theoretical review of the mathematical and physical background underlying the AlgebrasTQD code. The explicit code installation and examples will be described in Secs.~\ref{sec:install} and~\ref{sec:examples}.

\subsection{Twisted Quantum Doubles}
\label{sec:TQD}

The twisted quantum double $D^\omega(G)$ is the topological order of the $G$ gauge theory in 2+1d with Dijkgraaf-Witten twist $[\omega]\in H^3(G,U(1))$~\cite{Dijkgraaf:1989pz,Roche:1990hs,Hu:2012wx} and is the Symmetry Topological Field Theory (SymTFT) for every 1+1d system with symmetry related by gauging to the group $G$ with anomaly $\om$  \cite{Ji:2019jhk,Kong:2020cie,Gaiotto:2020iye,Apruzzi:2021nmk,Freed:2022qnc,Kaidi:2022cpf,Chatterjee:2022kxb,Chatterjee:2022jll,Moradi:2022lqp,Bhardwaj:2023fca,Bhardwaj:2023ayw,Bhardwaj:2023bbf,Bhardwaj:2024qrf,Bhardwaj:2024wlr,Bhardwaj:2024kvy,Chatterjee:2024ych,Warman:2024lir,Bottini:2025hri}, and its lattice realizations are the Kitaev quantum double model~\cite{Kitaev:1997wr} for trivial $\omega$ and the twisted quantum double model~\cite{Hu:2012wx,Mesaros:2012yd} in general.
Its representations are topological anyon worldlines and form the Modular Tensor Category given by the Drinfeld center $\cZ(\Vec_G^\omega)\simeq\Rep(D^\omega(G))$ of the fusion category $\Vec_G^\omega$. 

\paragraf{Group 3-cocycles.} 
Group $3$-cocycles can be computed by \texttt{omf} or, up-to-automorphism, by \texttt{omaut}, which use HAP~\cite{HAP} and some functions from \cite{mignard2017moritaequivalencepointedfusion}.
AlgebrasTQD stores each 3-cocycle in the returned list as a GAP record \texttt{rec(a,N)} in terms of an integer-valued function $a$ and a modulus $N$, with 
\be
\omega(g_1,g_2,g_3)=e^{2\pi i\,a(g_1,g_2,g_3)/N}\,.
\ee
They are normalized representatives of the classes of $H^3(G,U(1))$, and \texttt{3coctest} checks the normalization and the cocycle condition
\be
    \omega(g_2,g_3,g_4)\,\omega(g_1,g_2g_3,g_4)\,\omega(g_1,g_2,g_3)=\omega(g_1g_2,g_3,g_4)\,\omega(g_1,g_2,g_3g_4)\,.
    \label{eq:cocycle}
\ee
(using the additive function $a$ modulo $N$). The transgression and the tensor-product phase of $\omega$~\cite{Roche:1990hs,Gai:2026hjk}, computed by \texttt{3cocr}, \texttt{3cocl}, are, respectively
\be
    \omega(f,g|h)=\frac{\omega(f,{}^{g}h,g)}{\omega(f,g,h)\,\omega({}^{fg}h,f,g)}\,,\qquad \Omega_f(g,h)=\omega(f|g,h)=\frac{\omega(f,g,h)\,\omega({}^{f}g,{}^{f}h,f)}{\omega({}^{f}g,f,h)}\,.
    \label{eq:transgression}
\ee

\paragraf{Anyons.} The anyons are the topological lines in $D^\omega(G)$, and can be labeled by pairs $([g],\rho)$, in terms of a conjugacy class $[g]$, also known as the $G$ magnetic flux, and an irreducible representation $\rho$ of the centralizer $C_G(g)$ with 
\be
\rho(h_1)\rho(h_2)=\omega(h_1,h_2|g)^{-1}\rho(h_1h_2)\,,
\ee
also called the $G$ electric charge~\cite{Coste:2000tq}. The code of Gruen and Morrison~\cite{gruen2021computing} computes the projective characters $\rho$, and the function \texttt{Anyons} of AlgebrasTQD labels every anyon by its flux and by values of its character, as 
\be
([g],\chi_\rho(g)\ k{=}\chi_\rho(k)\,\dots)\,,
\ee 
for elements $k$ of $C_G(g)$ chosen so that all labels are distinct, or, for untwisted Abelian groups, by its magnetic and electric charges $m^ie^j$. \texttt{ShowAnyons} prints the labels with the quantum dimensions and topological spins,
\be
    d_{([g],\rho)}=|[g]|\,\chi_\rho(1)\,,\qquad \theta_{([g],\rho)}=\frac{\chi_\rho(g)}{\chi_\rho(1)}\,.
\ee
The pure charges $([1],R)$, with $R$ an irreducible representation of $G$, and the pure fluxes $([g],1)$ are the special cases, and the braiding of a charge around a flux is the Aharonov--Bohm effect of the gauge theory: shrinking a loop of $R$ around a line of $[g]$ produces the factor $\ol{\chi_R(g)}$, a phase when $R$ is one-dimensional, and a twist in the line of an anyon $a$ determines its spin $\theta_a$,
\be \label{eqfig1}
    \vcenter{\hbox{\begin{tikzpicture}[scale=0.8]
        \draw[thick,red] (0.8,0) arc (0:180:0.8 and 0.3);
        \draw[line width=4pt,white] (0,-1.1) -- (0,1.1);
        \draw[thick,blue,->] (0,-1.1) -- (0,1.1) node[above] {$[g]$};
        \draw[line width=4pt,white] (-0.8,0) arc (180:360:0.8 and 0.3);
        \draw[thick,red,->] (0.8,0) arc (360:240:0.8 and 0.3);
        \draw[thick,red] (-0.8,0) arc (180:360:0.8 and 0.3) node[right] {$R$};
    \end{tikzpicture}}}
    =\ol{\chi_R(g)}\;
    \vcenter{\hbox{\begin{tikzpicture}[scale=0.8]
        \draw[thick,blue,->] (0,-1.1) -- (0,1.1) node[above] {$[g]$};
    \end{tikzpicture}}}\,,\qquad\qquad
    \vcenter{\hbox{\begin{tikzpicture}[scale=0.8]
        \draw[thick,blue,->] (0.8,0) .. controls (0.8,-0.35) and (0.3,-0.3) .. (0,0) .. controls (-0.25,0.25) and (0,0.5) .. (0,1.1) node[above] {$a$};
        \draw[line width=4pt,white] (0,-1.1) .. controls (0,-0.5) and (-0.25,-0.25) .. (0,0) .. controls (0.3,0.3) and (0.8,0.35) .. (0.8,0);
        \draw[thick,blue] (0,-1.1) .. controls (0,-0.5) and (-0.25,-0.25) .. (0,0) .. controls (0.3,0.3) and (0.8,0.35) .. (0.8,0);
    \end{tikzpicture}}}
    =\theta_a\;
    \vcenter{\hbox{\begin{tikzpicture}[scale=0.8]
        \draw[thick,blue,->] (0,-1.1) -- (0,1.1) node[above] {$a$};
    \end{tikzpicture}}}\,.
\ee

\paragraf{Twisted Quantum Double.} \texttt{TQDTheory} returns the labeled anyons with the modular $S$ and $T$ matrices~\cite{Coste:2000tq}, with $S$ normalized as in Figure~\ref{fig:hopf}, computed with the code of Gruen and Morrison~\cite{gruen2021computing}, or by closed formulas for untwisted Abelian groups, and \texttt{Fusion} computes the fusion rules from $S$ by the Verlinde formula $N_{ab}^c=\sum_x S_{ax}S_{bx}\ol{S_{cx}}/S_{1x}$. The modular data does not fully determine the theory, since inequivalent modular tensor categories can have the same $S$ and $T$~\cite{mignard2021modular,Delaney:2021owx,Ng:2023wsc}, already among the twisted quantum doubles $\cZ(\Vec_G^\omega)$~\cite{mignard2021modular}, and condensable algebras with the same anyons can be inequivalent~\cite{Gai:2026hjk}. For this reason, AlgebrasTQD uses as inputs to \texttt{TQDTheory} the group and 3-cocycle data $(G,\omega)$ and does not perform computations from $S$ and $T$ alone.

\begin{figure}
\centering
\begin{tikzpicture}[baseline=(base)]
    \coordinate (base) at (0,-0.1);
    \draw[thick,blue] (0,0) ++(60:0.8) arc (60:402:0.8);
    \node[blue] at (-1.05,0) {$a$};
    \draw[thick,red] (1,0) circle (0.8);
    \node[red] at (2.05,0) {$b$};
    \draw[line width=4pt,white] (0,0) ++(-72:0.8) arc (-72:-30:0.8);
    \draw[thick,blue] (0,0) ++(-80:0.8) arc (-80:-25:0.8);
    \draw[thick,blue,->] (0,0) ++(-20:0.8) arc (-20:5:0.8);
    \draw[thick,red,->] (1,0) ++(-20:0.8) arc (-20:5:0.8);
\end{tikzpicture}
$\;=\;|G|\,S_{ab}\,,\qquad\qquad$
\begin{tikzpicture}[baseline=(base)]
    \coordinate (base) at (0,-0.1);
    \draw[thick,blue] (0,0) circle (0.8);
    \node[blue] at (-1.05,0) {$a$};
\end{tikzpicture}
$\;=\;d_a=|G|\,S_{1a}\,.$
\caption{The modular $S$ matrix of $\cZ(\Vec_G^\omega)$: the Hopf link of two anyons $a$ and $b$ is $\mathcal{D}\,S_{ab}$, with $\mathcal{D}=\bigl(\sum_ad_a^2\bigr)^{1/2}=|G|$ the total quantum dimension, and a single loop of $a$ is its quantum dimension $d_a=\mathcal{D}\,S_{1a}$. \texttt{TQDTheory} returns $S$ in this normalization, symmetric and unitary, and \texttt{Fusion} uses it in the Verlinde formula. The charge loop around a flux line in \eqref{eqfig1} is the case $a=([g],1)$, $b=([1],R)$.}
\label{fig:hopf}
\end{figure}

\subsection{Condensable Algebras}
\label{sec:algebras}

A condensable algebra is a set of anyons that can be condensed consistently, together with the multiplication that says how they fuse into each other. Mathematically, it is a connected, commutative, symmetric, normalized-special Frobenius algebra in a braided fusion category~\cite{Fuchs:2002cm,Frohlich:2003hm,Kong:2013aya}. It is called Lagrangian when it is maximal and its reduced topological order, the category of its local modules, is trivial~\cite{Davydov:2010kfz}. The condensable algebras in $\cZ(\Vec_G^\omega)$ are classified by group-theoretical data $A(H,F,\gamma,\eps)$~\cite{Davydov2009ModularIF,davydov2017lagrangian,Hannah:2023tae,Gai:2026hjk}, which we now review.

\paragraf{Algebra data.} A condensable algebra 
\be \label{eq:algebra}
A(H,F,\gamma,\eps)
\ee
is classified by a subgroup $H\leq G$, a normal subgroup $F\triangleleft H$, a normalized $\gamma:F\times F\to U(1)$ with $d\gamma=\omega|_{F\times F\times F}$, and $\eps:H\times F\to U(1)$ satisfying, for all $h,h_1,h_2\in H$ and $f,f_1,f_2\in F$, the conditions of the classification of~\cite{davydov2017lagrangian} in the form of~\cite[App.~A]{Gai:2026hjk}. In AlgebrasTQD,
\begin{enumerate}
    \item \texttt{AlgSubgroups} lists the pairs $F\triangleleft H$ up to conjugacy, 
    \item \texttt{gammaf} returns one $\gamma$ solving \eqref{eq:dgamma} per class of $H^2(F,U(1))$, and \texttt{gammatest} checks it,  
    \item \texttt{epsf} constructs all inequivalent solutions to the relations \eqref{eq:conditions},
    \item \texttt{AlgebraData} combines them.
\end{enumerate}
Physically, as a $\Vec_G^\omega$-symmetric phase $A(H,F,\gamma,\eps)$ describes a phase whose preserved symmetry is $H$, and $F$ labels the support of its twisted sectors.
$\gamma$ enters the multiplication of the twisted-sector operators $\cO_f$ for $f\in F$
\be
        \cO_{f_1}\cO_{f_2}=\gamma(f_1,f_2)\, \cO_{f_1f_2} 
\ee
while $\eps$ is a phase from the action of the symmetry on them \cite{Wen:2023otf,Gai:2026hjk}:
\be
h\,\cO_fh^{-1}=\eps(h,f)\, \cO_{hfh^{-1}}\,.
\ee
The dimension of an algebra is
\be
\dim \left( A(H,F,\gamma,\eps) \right) = |G|\,|F|/|H|\,,
\ee
and the algebra is Lagrangian if and only if $H=F$, in which case $\eps$ is fixed by $\gamma$ and we write $\cL(H,\gamma)$~\cite{davydov2017lagrangian,Gai:2026hjk}.

The algebra $A(H,F,\gamma,\eps)$ is spanned by $v_{g,f}$, $g\in G$, $f\in F$, where $f$ labels the twisted sector and the $v_{r,f}$ with $r$ in a set $R$ of representatives of the left cosets $rH$ form a basis~\cite{Davydov2009ModularIF,Hannah:2023tae,Gai:2026hjk}. \texttt{OperatorAlgebra} returns the algebra in this basis, with the grading, the multiplication and the $G$-action
\be\ba
    |v_{r,f}|&=rfr^{-1}\,, &\quad 
    v_{r,f}\,v_{r',f'}&=\delta_{r,r'}\,\frac{\gamma(f,f')}{\Omega_r(f,f')}\,v_{r,ff'}\,, &\quad
    g\cdot v_{r,f}&=\frac{\eps(h,f)\,\omega(r',h|f)}{\omega(g,r|f)}\,v_{r',{}^{h}f}\,,
    \label{eq:structure}
\ea\ee
for $gr=r'h\,,\; r'\in R\,,\; h\in H$. 

\paragraf{Anyon Decomposition.} The condensable algebra $A(H, F, \gamma, \epsilon)$ in $\cZ(\Vec_G^\omega)$ has the decomposition:
\be 
    A(H,F,\gamma,\epsilon) \cong_{\text{obj}} \bigoplus_{([x],\rho)} n_{([x],\rho)} ([x],\rho),
\ee
with the multiplicity $n_{([x],\rho)}$ of each simple object $([x], \rho)$ proven analytically in \cite{Gai:2026hjk}
\be\ba
    n_{([x],\rho)}&=\frac{1}{|G|}\sum_{\substack{f\in[x]\\ g\in C_G(f)}}\ol{\Tr\rho_f(g)}\,\sum_{\substack{y\in R\\ f^y\in F,\ g^y\in H}}\eps(g^y,f^y)\,\frac{\omega(g^y,y^{-1}|f)}{\omega(y^{-1},g|f)}
    \label{eq:anyons}
\ea\ee
with $x^y=y^{-1}xy$.
\texttt{AlgAnyonsV} and \texttt{AlgAnyonsProducts} evaluate it and check that every $n_a$ is a non-negative integer and that the quantum dimensions agree $\sum_a n_ad_a=|G|\,|F|/|H|$. 

\paragraf{Isomorphism Classes of Algebras.} \texttt{AlgebraClassesG} sorts the algebras into isomorphism classes according to the relations in \cite{Gai:2026hjk}: $A(H,F,\gamma,\eps)$ and $A(H',F',\gamma',\eps')$ are isomorphic if and only if there are $g\in G$ with $H'={}^{g}H$ and $F'={}^{g}F$, and $c:F\to U(1)$, with
\be\ba
    \frac{\gamma'({}^{g}f_1,{}^{g}f_2)\,\Omega_g(f_1,f_2)}{\gamma(f_1,f_2)}&=\frac{c(f_1f_2)}{c(f_1)\,c(f_2)}\,, &\qquad
    \frac{\eps'({}^{g}h,{}^{g}f)\,\Omega_g(h,f)}{\eps(h,f)\,\Omega_g({}^{h}f,h)}&=\frac{c({}^{h}f)}{c(f)}\,,
    \label{eq:isomorphism}
\ea\ee
for all $h\in H$ and $f,f_1,f_2\in F$. \texttt{AlgebraClasses} attaches to each class its name, dimension and anyons.

\paragraf{Twin Algebras.} Two algebras are twins if they have the same anyon decomposition but are not isomorphic~\cite{Gai:2026hjk}. Twins come in four types: with non-conjugate $H$,\footnote{For example from a Gassmann triple $(G,H_1,H_2)$~\cite{Gassmann:1926,Sunada1985RiemannianCA}.} with $F$ not conjugate under the normalizer $N_G(H)$, with inequivalent $\gamma$, for which Bogomolov multipliers give the Lagrangian examples~\cite{FABogomolov_1988,Davydov:2013xov,Pollmann:2012,Kobayashi:2025ykb,Kobayashi:2025pxs}, or with inequivalent $\eps$~\cite{Gai:2026hjk}. \texttt{AlgebraClassesG} determines them and stores the type of every twin pair it finds.

\paragraf{Subalgebras and Phases.} The condensable algebras are partially ordered by inclusion, with $A(H,F,\gamma,\eps)\preceq A(H',F',\gamma',\eps')$ generated, up to isomorphism, by the two relations~\cite{Bhardwaj:2024qrf,Gai:2026hjk}
\be\ba
    &H=H'\,,\quad F<F'\,,\quad \gamma=\gamma'|_{F\times F}\,,\quad \eps=\eps'|_{H\times F}\,, \cr
    &H'<H\,,\quad F=F'\,,\quad \gamma=\gamma'\,,\quad \eps'=\eps|_{H'\times F}\,.
    \label{eq:order}
\ea\ee
\texttt{IsSubalgebra} decides this order between two isomorphism classes of algebras and \texttt{Hasse} arranges the classes by dimension into the Hasse diagram \cite{Bhardwaj:2024qrf}, with $A(G,1,1,1)\cong_{\text{obj}} 1$ at the top and the Lagrangian algebras at the bottom.

For a symmetry boundary $\cL_\sym$, a Lagrangian algebra, the phase obtained from the SymTFT with the algebra $A$ on the physical boundary has $n=\sum_an_a(\cL_\sym)\,n_a(A)$ vacua, or gapless universes~\cite{Bhardwaj:2024qrf}. \texttt{PhaseLabels} names the phase SPT ($n=1$) or SSB ($n>1$) when $A$ is Lagrangian, and gSPT or gSSB when it is not, with the prefix i, for ``intrinsically gapless'', when every Lagrangian algebra that contains $A$ has more than $n$ vacua~\cite{Bhardwaj:2024qrf}. 

\paragraf{Reduced topological order.} An algebra $A(H,F,\gamma,\eps)$ defines a topological interface $\cI_{A(H,F,\gamma,\eps)}$ from $D^\omega(G)$ to its reduced topological order $D^\alpha(H/F)$, as shown in Fig.~\ref{fig:symtft}, since the category of local modules of $A(H,F,\gamma,\eps)$ is $\cZ(\Vec_{H/F}^\alpha)$. Some physics applications are discussed in e.g.~\cite{Beigi:2010htr,Chatterjee:2022tyg,Bhardwaj:2023bbf,Bhardwaj:2024qrf,Bhardwaj:2025jtf}.
The expression for the 3-cocycle $\alpha$ is provided in \cite{Gai:2026hjk}. In AlgebrasTQD, \texttt{alphaf} computes it, verifies with \texttt{3coctestm} that the result is a normalized $3$-cocycle, and stores its conjugate $\bar\alpha$. The pullback of $\alpha$ to $H$ is cohomologous to $\omega|_H$~\cite{Hannah:2023tae}.

Folding along the interface associates to an interface a Lagrangian algebra of the folded theory. The decomposition of the folded Lagrangian algebra into anyons gives the map of anyons from the reduced topological order to the original one,
\be\ba
    \cL^\folded(H^{\diag},\varphi)&=\bigoplus_{a\in\cZ(\Vec_G^\omega),\ b\in\cZ(\Vec_{H/F}^\alpha)}n_{ab}\,a\otimes\bar b\,, \cr
    \cZ(\Vec_{H/F}^\alpha)&\to\cZ(\Vec_G^\omega)\,,\qquad b\mapsto\bigoplus_a n_{ab}\,a\,,
    \label{eq:map}
\ea\ee
which sends the vacuum to $A(H,F,\gamma,\eps)$, confines every anyon that is not in its image, for example every $([g],\rho)$ with $[g]\cap H=\emptyset$, and is compatible with the modular data, $Sn=nS'$ and $Tn=nT'$~\cite{Chatterjee:2022tyg,Ng:2023wsc,Gai:2026hjk}. \texttt{InterfaceMap} computes one of these possible maps. Two folded Lagrangian algebras are related by a braided autoequivalence of the reduced topological order: \texttt{AllInterfaceMaps} composes the map of \texttt{InterfaceMap} with the braided autoequivalences computed by \texttt{AutomorphismsAnyons} to return all of them.

\subsection{Computing the Solutions}
\label{sec:solver}

To determine a trivialization $\gamma$ of $\omega|_F$, a function $\eps$ satisfying \eqref{eq:conditions}, or a gauge transformation $c$ used in the isomorphism relation \eqref{eq:isomorphism}, the AlgebrasTQD code uses an exact linear solver. It passes from multiplicative to additive cochains to solve linear equations, whose number it furthermore reduces by working with subgroup generators, as proved in App.~\ref{app:proofs}.

\paragraf{Additive and Multiplicative Cochains.} In AlgebrasTQD, a $3$-cocycle is stored additively, as the record \texttt{rec(a,N)} with $\omega=e^{2\pi i\,a/N}$, while the cochains $\gamma$ and $\eps$ and the phases \eqref{eq:transgression} are stored multiplicatively, as GAP functions that return cyclotomic numbers $\texttt{E(n)}^k=e^{2\pi ik/n}$, which are multiplied and compared exactly. The two forms are related by the argument of a root of unity, $\langle e^{2\pi ik/n}\rangle=k/n\in\mathbb{Q}/\Z$, which GAP returns through \texttt{DescriptionOfRootOfUnity}. An unknown cochain is written $c(f)=e^{2\pi i\,x_f}$ with $x_f\in\mathbb{Q}/\Z$, and a multiplicative equation with integer exponents $m_i$ becomes a linear congruence,
\be
    \prod_i c(f_i)^{m_i}=z\qquad\Longleftrightarrow\qquad \sum_i m_i\,x_{f_i}\equiv\langle z\rangle\bmod 1\,.
    \label{eq:additive}
\ee
A system of such equations is $Ax\equiv b\bmod 1$ with $A$ an integer matrix and $b$ a rational vector, and a solution $x_f=p/q$ is returned to multiplicative form as $c(f)=\texttt{E(q)}^p$. Conversely, \texttt{alphaf} evaluates the twist of the reduced topological order multiplicatively and then stores it as an additive cocycle \texttt{rec(a,N)}, with $N$ the least common multiple of the orders of its values, so that the reduced topological order is built by \texttt{TQDTheory}.

\paragraf{The Linear Solver.} \texttt{SolveModZ} solves $Ax\equiv b\bmod 1$, with $A\in\Z^{r\times m}$, by the Hermite normal form. Denoting by $D$ the common denominator of $b$, then, \texttt{HermiteNormalFormIntegerMat} gives
\be
    U\,[\,A\mid Db\,]=[\,H\mid b'\,]\,,
    \label{eq:hnf}
\ee
with $U$ a unimodular integer matrix and $H$ in row echelon form. Since $U$ and $U^{-1}$ have integer entries, $U$ maps $D\Z^r$ onto itself, and for every $x\in\mathbb{R}^m$
\be
    Ax-b\in\Z^r\quad\Longleftrightarrow\quad H(Dx)-b'\in D\Z^r\,.
\ee
\texttt{SolveModZ} reduces $b'$ modulo $D$ to a vector $v$ and solves $Hy=v$ over $\mathbb{Q}$ with \texttt{SolutionMat}. If $Ax\equiv b\bmod1$ has a solution, the rows $i$ in which $H$ vanishes give $b'_i\in D\Z$, that is $v_i=0$. Conversely, if $v_i=0$ for these rows, $Hy=v$ has a rational solution, since the non-zero rows of a matrix in row echelon form are linearly independent, and $x=y/D$ satisfies $H(Dx)-b'=v-b'\in D\Z^r$. Hence \texttt{SolutionMat} returns \texttt{fail} exactly when there is no solution, and otherwise \texttt{SolveModZ} returns $x=y/D$, after checking $Ax-b\in\Z^r$. Since the equivalence above holds for every real $x$, that is for every unknown phase $e^{2\pi ix}\in U(1)$, an output of \texttt{fail} proves that the multiplicative equations have no solution. Otherwise, a solution in roots of unity is returned.

\paragraf{Trivializing $\omega$ on a generating set.} \texttt{GammaTrivialization} solves $d\gamma=\omega|_{F\times F\times F}$, i.e.
\be
    \gamma(f_2,f_3)\,\gamma(f_1,f_2f_3)=\gamma(f_1f_2,f_3)\,\gamma(f_1,f_2)\,\omega(f_1,f_2,f_3)\,,
    \label{eq:dgamma}
\ee
with unknowns $\gamma(f_1,f_2)$ for $f_1,f_2\neq1$, by imposing the equation at the triples $(f_1,f_2,s)$ with $f_1,f_2\neq1$ and $s$ in a generating set $S$ of $F$, that is $(|F|-1)^2|S|$ equations. Theorem~\ref{thm:gamma} proves that these equations imply \eqref{eq:dgamma} on all of $F^3$, and that a normalized solution exists if and only if $\omega|_F$ is a coboundary. Every other trivialization is $\gamma\beta$ with $\beta$ a $2$-cocycle, and \texttt{gammaf} returns one per class of $H^2(F,U(1))$, since trivializations that differ by a coboundary give isomorphic algebras \cite{Ostrikmodule,Natale2017,Gai:2026hjk}. \texttt{GaugeableSubgroups} lists the subgroups $F$ on which $\omega$ is trivial, up to conjugacy.

\paragraf{Solving for $\eps$.} For fixed $(H,F,\gamma)$ and a generating set $S_H$ of $H$, \texttt{epsf} constructs every $\eps$ in four steps, by solving \cite{davydov2017lagrangian,Gai:2026hjk}
\be\ba
    \eps(h_1h_2,f)&=\eps(h_1,{}^{h_2}f)\,\eps(h_2,f)\,\omega(h_1,h_2|f)\,, \cr
    \gamma(f_1,f_2)\,\eps(h,f_1f_2)&=\eps(h,f_1)\,\eps(h,f_2)\,\gamma({}^{h}f_1,{}^{h}f_2)\,\Omega_h(f_1,f_2)\,, \cr
    \eps(f_1,f_2)&=\frac{\gamma(f_1,f_2)}{\gamma({}^{f_1}f_2,f_1)}\,.
    \label{eq:conditions}
\ea\ee
\begin{enumerate}
\item On $F\times F$, $\eps$ is fixed by $\gamma$ through the third line of \eqref{eq:conditions}.
\item For $s\in S_H$ not in $F$, the second line of \eqref{eq:conditions} says that $\eps(s,\cdot)$ is a projective character of $F$,
\be
    \eps(s,f_1)\,\eps(s,f_2)=\beta_s(f_1,f_2)\,\eps(s,f_1f_2)\,,\qquad \beta_s(f_1,f_2)=\frac{\gamma(f_1,f_2)}{\Omega_s(f_1,f_2)\,\gamma({}^sf_1,{}^sf_2)}\,.
    \label{eq:projective}
\ee
\texttt{SolveModZ} finds one solution $k_s$, or proves that there is none, and the candidates are $\eps(s,\cdot)=k_s\chi_s$ with $\chi_s$ a linear character of $F$.
\item Each candidate is extended to all of $H$ by the first line of \eqref{eq:conditions},
\be
    \eps(hs,f)=\eps(h,{}^{s}f)\,\eps(s,f)\,\omega(h,s|f)\,,\qquad s\in S_H\,.
    \label{eq:propagate}
\ee
Whenever this assigns a second value to an element whose value is already known, the code checks that the two values agree, (otherwise, the candidate $\eps$ is discarded since it is inconsistent).
\item The linear characters $\chi$ of $F$ act on the solutions by $\eps(h,f)\mapsto\eps(h,f)\,\chi({}^hf)/\chi(f)$, which is the isomorphism relation \eqref{eq:isomorphism} of~\cite{Davydov2009ModularIF,davydov2017lagrangian,Gai:2026hjk} with $g=1$ and $dc=1$. \texttt{epsf} keeps one candidate per class and tests it against the first two lines of \eqref{eq:conditions} on all of $H\times H\times F$ and $H\times F\times F$.
\end{enumerate}
Theorem~\ref{thm:eps} proves that the output is exactly one $\eps$ per algebra class of the solutions of \eqref{eq:conditions}. For Abelian $H$, \texttt{epsfA} follows the same steps with independent generators and trivial conjugation action.

\paragraf{Algebra Isomorphisms and Inclusions.}  These equations and the partial order relations \eqref{eq:order} use similar algorithms in AlgebrasTQD. For every $g$ that conjugates $(H,F)$ to $(H',F')$, \texttt{AlgebraClassesG} solves the first equation of \eqref{eq:isomorphism} for $c$ with \texttt{IsTrivial2coc}, through \texttt{SolveModZ}, and then tests the second for $c$ multiplied by every linear character of $F$, which is complete by Lemma~\ref{lem:isotest}. Similarly, for every group element that conjugates the subgroups into each other, \texttt{IsSubalgebra} solves both conditions \eqref{eq:order} on the restricted $\gamma$ and $\eps$ data.

\paragraf{Correctness Tests.} Every quantity that AlgebrasTQD computes is checked when it is produced. For example, \texttt{3coctest} and \texttt{gammatest} verify the cocycle and trivialization conditions, \texttt{epsf} returns only solutions of \eqref{eq:conditions}, \texttt{AlgAnyonsV} checks that the multiplicities \eqref{eq:anyons} are non-negative integers with $\sum_an_ad_a=\dim A$, and \texttt{alphaf} checks that the twist of the reduced topological order is a normalized $3$-cocycle.

\section{Installation and Source Files}
\label{sec:install}
In this section, we first explain how to install GAP and the AlgebrasTQD code. We then summarize its source code.
 
\subsection{Installing GAP and AlgebrasTQD} 
To install GAP, the user can open the official installation GAP page \url{https://www.gap-system.org/install}, click the link for their operating system and follow the installation instructions. 
Then, the user can download the zip archive of AlgebrasTQD from \url{https://github.com/AlisonWarman/AlgebrasTQD} and unzip it
\begin{code}
unzip AlgebrasTQD-main.zip    
\end{code}
then change directory into it
\begin{code}
cd AlgebrasTQD-main
\end{code}
GAP is started from this directory by running \texttt{gap}. If GAP is installed but this command does not start it, users can create a symbolic link in this directory. For example, for gap-4.16.1.tar.gz unpacked under /home/user/.gap in Linux, the user can run:
\begin{code}
ln -s /home/user/.gap/gap-4.16.1/gap gap
\end{code}
GAP can then be started from the terminal in the AlgebrasTQD-main directory by running:
\begin{code}
./gap
\end{code}
Then, the user should install the required GAP packages with:\footnote{We show the GAP logo on the top right for code to be executed in a GAP interactive session.}
\begin{gapcode}
LoadPackage("PackageManager");
InstallPackage("HAP");
InstallPackage("IO");
InstallPackage("CTblLib");
\end{gapcode}
and load the package by running
\begin{gapcode}
Read("AlgebrasTQD.g");
\end{gapcode}
We will illustrate the theory and examples in detail in the following sections.\\
In GAP, every number is expressed exactly, using the roots of unity \texttt{E(n)}$=e^{2\pi i/n}$. \\
At the end of the session, to exit from GAP, the user can run
\begin{gapcode}
quit;
\end{gapcode}

\subsection{Source Files}
\label{sec:source}
We explain how to install GAP, its required packages and load AlgebrasTQD \cite{AlgebrasTQD} in Sec.~\ref{sec:install}. The main file \texttt{AlgebrasTQD.g} checks that the required packages (HAP, CTblLib and IO) are installed, declares the names used across files, and reads the eight source files of Table~\ref{tab:files} in order. The header of each source file lists its contents and above each function there is a  comment with its role, input and output.

\newpage
\begin{table}
\centering
\small
\renewcommand{\arraystretch}{1.3}
\begin{tabular}{|l|m{0.58\textwidth}|}
\hline
File in \texttt{source/} & Contents and main functions \\
\hline\hline
\makecell[l]{\texttt{0\_Modular\_Data\_Mignard-}\\\texttt{Schauenburg-Gruen-Morrison.g}} & 3-cocycles up to automorphism, projective characters of centralizers and the modular data of twisted quantum doubles~\cite{mignard2017moritaequivalencepointedfusion,gruen2021computing}. \\
\hline
\texttt{1\_preliminaries.g} & Groups with labeled generators from a name or an identifier (\texttt{MakeGroup}, \texttt{RelabelGroup}, \texttt{GroupElement}, \texttt{ShowGroup}), and the exact solver of linear systems modulo $1$ using the Hermite normal form (\texttt{SolveModZ}). \\
\hline
\texttt{2\_anyons.g} & Labels for all the anyons, their characters, quantum dimensions and topological spins, and the fusion rules (\texttt{Anyons}, \texttt{ShowAnyons}, \texttt{AbelianLabels}, \texttt{CustomAnyonLabels}, \texttt{Fusion}, \texttt{ShowFusion}). \\
\hline
\texttt{3\_cocycles.g} & $3$-cocycles using HAP \cite{HAP} (\texttt{omf}, \texttt{omaut}, \texttt{H3Invariants}, \texttt{H3Cocycle}, \texttt{3cocr}, \texttt{3cocl}), their tests and normalization (\texttt{3coctest}, \texttt{NormalizeCocycle}), the trivializations $\gamma$ and the subgroups admitting them (\texttt{gammaf}, \texttt{GaugeableSubgroups}), and the TQD topological theory as a GAP record (\texttt{TQDTheory}, \texttt{ShowTQDTheory}). \\
\hline
\texttt{4\_algebra\_data.g} & Subgroup pairs up to conjugacy (\texttt{AlgSubgroups}), the solutions $\eps$ of \eqref{eq:conditions} (\texttt{epsf}, \texttt{epsfA}), every condensable algebra of a TQD theory (\texttt{AlgebraData}, \texttt{TQDAlgebras}), names and display (\texttt{AlgebraName}, \texttt{ShowAlgData}, \texttt{AlgebraDisplayData}). \\
\hline
\texttt{5\_interface\_map.g} & The anyon decomposition of algebras \eqref{eq:anyons} (\texttt{AlgAnyonsV}, \texttt{AlgAnyonsProducts}, \texttt{AlgebrasWithAnyons}, \texttt{FindAlgebras}), the twist of the reduced topological order (\texttt{alphaf}), the interface maps \eqref{eq:map} (\texttt{InterfaceMap}, \texttt{AllInterfaceMaps}) and the braided autoequivalences (\texttt{AutomorphismsAnyons}, \texttt{AutomorphismsDG}). \\
\hline
\makecell[l]{\texttt{6\_algebra\_classes\_and\_}\\\texttt{inclusions.g}} & Isomorphism classes and twins by \eqref{eq:isomorphism} (\texttt{AlgebraClasses}), subalgebras and the Hasse diagram (\texttt{IsSubalgebra}, \texttt{Hasse}), phase names (\texttt{PhaseLabels}), the algebra of order parameters \eqref{eq:structure} (\texttt{OperatorAlgebra}), the topological sectors of a gapped boundary~\cite{Ostrikmodule} (\texttt{TopologicalSectors}), and the \LaTeX{} output functions (\texttt{AlgebraTableTeX}, \texttt{TwinTableTeX}, \texttt{HasseTikz}, \texttt{HasseTableTeX}). \\
\hline
\texttt{7\_quantum\_computing.g} & Qudit Pauli and Clifford groups and the Clifford hierarchy (\texttt{Dtilde}, \texttt{CliffordLevelOf}), vertex and gauging projectors (\texttt{VertexProjectorL}, \texttt{DiagProjector}), lattice surgery (\texttt{RoughMergeSplit}, \texttt{D4NSurgeryGate}), and SPT-stacking gates (\texttt{SPTStackingGate}, \texttt{D4NCocycle}, \texttt{D4NPhaseGate}). \\
\hline
\end{tabular}
\caption{The source files, read in this order by \texttt{AlgebrasTQD.g}.}
\label{tab:files}
\end{table}

\newpage
The basic data objects used by AlgebrasTQD are GAP lists and records. We summarize the main types here:
\begin{itemize}
\item An algebra $A(H,F,\gamma,\eps)$ is stored as the list \texttt{[H,F,gamma,eps,G]}, with \texttt{F} a normal subgroup of \texttt{H} and \texttt{gamma}, \texttt{eps} GAP functions of two group elements that return roots of unity.
\item A $3$-cocycle is stored as an additive record \texttt{rec(a,N)} with $\omega=e^{2\pi i\,a/N}$.
\item An anyon is stored as the list \texttt{[g,character,class,label]} of a representative $g$ of its conjugacy class, the character of its irreducible projective representation of the centralizer of $g$, stored as a dictionary from the elements of the centralizer to the character values, the conjugacy class of $g$, and the label of the anyon.
\end{itemize}
The data of $D^\omega(G)$ that the other functions take as input is:
\begin{itemize}
\item \texttt{G}: a finite group built by \texttt{MakeGroup}.
\item \texttt{glG}: the list \texttt{[elements,labels,gens]} that names every element of $G$ in terms of labeled generators.
\item \texttt{aG}: the list of anyons, whose order fixes the indexing of every vector and matrix over anyons.
\item \texttt{labels}: the names of the anyons as strings of characters, such as \texttt{m},\texttt{e} or \texttt{([a],E(4))}.
\item \texttt{H3}: the list of integers $[n_1,\dots,n_r]$ for $H^3(G,U(1))=\Z_{n_1}\times\dots\times\Z_{n_r}$, which fix the coordinates of a 3-cocycle class $[\omega]$.
\item \texttt{om}: the $3$-cocycle, bound only when it is non-trivial.
\item \texttt{S}, \texttt{T}: the modular matrices, indexed like \texttt{aG} and computed only when requested with \texttt{rec(ST:=true)}.
\end{itemize}
The algebra records returned by \texttt{AlgebrasWithAnyons} contain:
\begin{itemize}
\item \texttt{n}: the position of the algebra in the list.
\item \texttt{A}: the algebra as the list above.
\item \texttt{name}: a one-line name such as \texttt{A(S3=<r,s>, C3=<r>, 1, 1)}.
\item \texttt{H}, \texttt{F}: the two subgroups $H\subseteq G$ and $F\triangleleft H$.
\item \texttt{K}: the structure of the quotient $H/F$, the group of the reduced topological order.
\item \texttt{dim}: the dimension $\dim A=|G|\,|F|/|H|$.
\item \texttt{Lag}: \texttt{true} for a Lagrangian algebra $A$ with $\dim A=|G|$.
\item \texttt{na}: the multiplicities $n_a$ of \eqref{eq:anyons}, indexed as the anyons \texttt{th.aG}.
\item \texttt{decomposition}: the same decomposition as a string of anyon labels.
\item \texttt{labels}: a copy of \texttt{th.labels}.
\end{itemize}
We will discuss the main code functions in the example section \ref{sec:examples}.

\section{Examples}
\label{sec:examples}
All AlgebrasTQD code files (both the source and examples) are documented at the top, and above each function there is a header with its purpose, inputs and outputs.
In this Section, we illustrate how AlgebrasTQD can be used with explicit example scripts. Each example can be used as a starting point for a new calculation by changing the group and cocycle data or the types of algebras. The AlgebrasTQD directory \texttt{examples/} contains the ten numbered scripts of Table~\ref{tab:examples} and the cocycle tables \texttt{G32-43-omega.g} and \texttt{GL23-om-data.g} for examples $7$ and $8$. 

\begin{table}
\centering
\small
\renewcommand{\arraystretch}{1.3}
\begin{tabular}{|l|m{0.56\textwidth}|c|}
\hline
Script in \texttt{examples/} & Calculation and main functions & Time \\
\hline\hline
\texttt{0\_load.g} & Loads the code, lists the main functions. & 1\,s \\
\hline
\texttt{1\_Abelian.g} & Toric code, double semion, twisted $\Z_2^2$, $D(\Z_4)$ \newline
(\texttt{TQDTheory}, \texttt{ShowAnyons}). & 1\,s \\
\hline
\texttt{2\_non-Abelian.g} & Anyons, $S$, $T$ and fusion of $D(S_3)$, $D(D_8)$ \newline
(\texttt{MakeGroup}, \texttt{ShowTQDTheory}, \texttt{Fusion}). & 2\,s \\
\hline
\texttt{3\_cocycles.g} & $H^2$, $H^3$, custom cocycles, trivializations \newline (\texttt{H3Invariants}, \texttt{omaut}, \texttt{3coctest}, \texttt{gammaf}). & 2\,s \\
\hline
\texttt{4\_algebras.g} & Every algebra with its anyons, $D^{\om_\III}(\Z_2^3)$ and $D(D_8)$ (\texttt{AlgebrasWithAnyons}, \texttt{ShowAlgData}, \texttt{AlgebraClasses}, \texttt{AlgAnyonsProducts}). & 4\,s \\
\hline
\texttt{5\_interface\_map.g} & Anyons and algebras of the toric code and the double semion, electric-magnetic duality, canonical maps of anyons of $D(A_4)$ (\texttt{InterfaceMap}, \texttt{AutomorphismsAnyons}). & 3\,s \\
\hline
\makecell[l]{\texttt{6\_algebra\_classes}\\\texttt{\_and\_inclusions.g}} & Classes, Hasse diagrams and phases of $D(A_4)$ and its twisted doubles (\texttt{AlgebraClasses}, \texttt{Hasse}, \texttt{PhaseLabels}). & 2\,s \\
\hline
\texttt{7\_twin\_algebras\_G32.g} & The twin algebras of $G_{32}$, with trivial and non-trivial $\omega$ \cite{Gai:2026hjk} \newline (\texttt{AlgebraClasses}, \texttt{Hasse}). & 3\,min \\
\hline
\texttt{8\_twin\_algebras\_GL23.g} & The twin phases of $GL(2,3)^\omega$ \cite{Warman:2026gfz} \newline (\texttt{AlgebraClasses}, \texttt{InterfaceMap}, \texttt{PhaseLabels}). & 2\,min \\
\hline
\makecell[l]{\texttt{9\_quantum\_computing}\\\texttt{\_D4N.g}} & $T^{1/N}$ gates of $D(D_{4N})$ codes \cite{Huang:2025ump,Warman:2025hov} \newline (\texttt{D4NSurgeryGate}, \texttt{SPTStackingGate}). & 2\,min \\
\hline
\end{tabular}
\caption{The example scripts with the calculation each performs, the main functions it uses, and its approximate run time on one core of a personal computer. The times of examples $7$ and $8$ include the optional calculations with a non-trivial cocycle and the twin algebra sweep for groups of order up to $48$.}
\label{tab:examples}
\end{table}

To begin, start GAP from the AlgebrasTQD-main directory as explained at the start of Sec.~\ref{sec:install}:
\begin{code}
./gap
\end{code}
Then the user can load AlgebrasTQD
\begin{gapcode}
Read("AlgebrasTQD.g");
\end{gapcode}
or run any line below for the corresponding example file:
\begin{gapcode}
Read("examples/0_load.g");
Read("examples/1_Abelian.g");
Read("examples/2_non-Abelian.g");
Read("examples/3_cocycles.g");
Read("examples/4_algebras.g");
Read("examples/5_interface_map.g");
Read("examples/6_algebra_classes_and_inclusions.g");
Read("examples/7_twin_algebras_G32.g");
Read("examples/8_twin_algebras_GL23.g");
Read("examples/9_quantum_computing_D4N.g");
\end{gapcode}
We will explain each example in detail below. In this section, all code blocks with the GAP logo are executable in GAP, while the grey blocks without it are code outputs.

One can first check that the package loads correctly by running
\begin{gapcode}
Read("examples/0_load.g");
\end{gapcode}
which prints
\begin{gapout}
AlgebrasTQD loaded. Main functions:
  MakeGroup
  MakeGroupLabels
  RelabelGroup
  GroupElement
  GroupLabel
  ShowGroup
  TQDTheory
  GaugeableSubgroups
  Fusion
  ShowFusion
  ShowAnyons
  AlgebrasWithAnyons
  ShowAlgebrasWithAnyons
  AlgebraClasses
  Hasse
  PhaseLabels
  InterfaceMap
  RoughMergeSplitStep
  RoughMergeSplit
  SPTStackingGate
\end{gapout}

\subsection{Abelian Quantum Doubles}
\label{sec:ex1}
Let's start with the simplest quantum doubles, those for Abelian groups, with the example script 
\begin{gapcode}
Read("examples/1_Abelian.g");
\end{gapcode}
A twisted quantum double topological theory is created by \texttt{TQDTheory} from a group, its name or ID in the small groups library, and a cocycle given by its coordinates in $H^3(G,U(1))$ or as a record \texttt{rec(a,N)}. We illustrate this for the toric code $D(\Z_2)$ and the double semion $D^\omega(\Z_2)$ theory (GAP calls the cyclic groups $\Z_n$ by $C_n$).
\begin{gapcode}
ShowAnyons(TQDTheory("C2"));         # toric code
ShowAnyons(TQDTheory("C2", [1]));    # double semion
\end{gapcode}
where for the double semion, \texttt{("C2", [1])} denotes the group $\Z_2$ with non-trivial 3-cocycle class (written additively here as $[1]$). Its output is
\begin{gapout}
Toric code D(C2)                 Double semion
1   1     d=1   theta=1          1   ([1],1)         d=1   theta=1
2   e     d=1   theta=1          2   ([1],1 a=-1)    d=1   theta=1
3   m     d=1   theta=1          3   ([a],E(4))      d=1   theta=E(4)
4   me    d=1   theta=-1         4   ([a],E(4)^3)    d=1   theta=-E(4)
\end{gapout}
Next, we turn to $D^{\om_{\text{II}}}(\Z_2\times\Z_2)\cong D(\Z_4)$: 
\begin{gapcode}
K := MakeGroup("C2 x C2");;
omII := First(omaut(K), om -> Length(GaugeableSubgroups(om)) = 3);;
ShowAnyons(TQDTheory(K, omII));      # type-II twisted C2 x C2
ShowAnyons(TQDTheory("C4"));
\end{gapcode}
\texttt{MakeGroup} attaches generator labels, $a$ and $b$ for $\Z_2\times\Z_2$, then \texttt{omaut} returns one 3-cocycle for each automorphism orbit of $H^3(\Z_2\times\Z_2,U(1))=\Z_2^3$ and \texttt{GaugeableSubgroups} the subgroups (up to conjugacy) on which a 3-cocycle is trivial. The type-II cocycle $\omega_{\rm II}(g_1,g_2,g_3)=(-1)^{x_1y_2y_3}$ for $g=a^xb^y$ is trivial on three subgroups of $\Z_2\times\Z_2$: $1$, $\langle a\rangle$ and $\langle b\rangle$. For an untwisted Abelian group the anyon labels are magnetic and electric, $m^ie^j$, and for a twisted theory an anyon $([g],\rho)$ is written by the values of $\rho$ on the generators of the centralizer. The type-II twisted $\Z_2\times\Z_2$ has sixteen Abelian anyons, whose twelve anyons with non-trivial flux have order four under fusion, as in $D(\Z_4)$, which the script prints last.

\subsection{Non-Abelian Quantum Doubles}
\label{sec:ex2}

The second example 
\begin{gapcode}
Read("examples/2_non-Abelian.g");
\end{gapcode}
uses non-Abelian groups. Let's start with $S_3$:
\begin{gapcode}
G := MakeGroup("S3");;
ShowGroup(G);
th := ShowTQDTheory("S3");;
Display(th.S);  Display(th.T);
Print(Fusion(th, "([s],1)", "([s],1)"), "\n");
\end{gapcode}
\texttt{MakeGroup("S3")} labels the rotation $r$ of order $3$ and the reflection $s$, of $S_3$, then \texttt{ShowGroup} prints the relations and the label of every element. \texttt{ShowTQDTheory} builds the twisted quantum double theory with its modular data and prints the anyons and their fusion rules.
The eight anyons of $D(S_3)$ are the three charges $([1],\rho)$ for $\rho\in\Irr(S_3)$, then two anyons with fluxes in the class of $s$ and  three with fluxes in the class of $r$, labeled by the conjugacy class and irrep of the centralizer ($\Z_2$ and $\Z_3$ respectively):
\begin{gapout}
1   ([1],1)        d=1   theta=1
2   ([1],1 s=-1)   d=1   theta=1
3   ([1],2 r=-1)   d=2   theta=1
4   ([s],1)        d=3   theta=1
5   ([s],-1)       d=3   theta=-1
6   ([r],1)        d=2   theta=1
7   ([r],E(3)^2)   d=2   theta=E(3)^2
8   ([r],E(3))     d=2   theta=E(3)
\end{gapout}
where $\texttt{([1],1 s=-1)}$ and $\texttt{([1],2\ r{=}{-}1)}$ are the sign and the two-dimensional representations of $S_3$ respectively.
\texttt{Fusion} takes two anyons by label or index and computes their product from $S$ by the Verlinde formula, e.g. here we compute the fusion of the anyon of dimension 3 and spin 1 $([s],1)$ with itself:
\begin{gapcode}
n := Length(th.aG);;
d := List(th.aG, a -> dima(a, One(G)));;
k := First([1..n], i -> d[i] = 3 and th.T[i][i] = 1);;
Print(th.labels[k], " x ", th.labels[k], " = ", Fusion(th, k, k), "\n\n");
\end{gapcode}
and the code prints
\begin{gapout}
([s],1) x ([s],1) = ([1],1)+([1],2 r=-1)+([r],1)+([r],E(3)^2)+([r],E(3))
\end{gapout} 
Next, we show how to re-label group elements of the dihedral group $D_8$ of order 8:
\begin{gapcode}
D8 := MakeGroup("D8", ["c","a","b"], ["s","sr","sr^3"]);;
ShowAnyons(TQDTheory(D8));
\end{gapcode}
This code labels $D_8=(\Z_2^a\times\Z_2^b)\rtimes\Z_2^c$ as in~\cite{Bhardwaj:2024qrf}, from its standard notation $D_8=\Z_4^r\rtimes\Z_2^s$ with rotation $r$ and reflection $s$. The two presentations are related by $c=s$, $a=sr$ and $b=sr^3$, and the code prints the $22$ anyons of $D(D_8)$. We will encounter a third presentation of $D(D_8)$, as $D^{\om_\III}(\Z_2^3)$, in Secs.~\ref{sec:ex3} and~\ref{sec:ex4}.

\subsection{Cocycles}
\label{sec:ex3}
The script 
\begin{gapcode}
Read("examples/3_cocycles.g");
\end{gapcode}
shows how the 3-cocycles $\omega$ can be computed by the code, or defined by the user and tested. 
\texttt{H3Invariants} returns $H^3(G,U(1))=\Z_{n_1}\times\dots\times\Z_{n_r}$ as the list $[n_1,\dots,n_r]$. For the actual 3-cocycle functions of 3 group elements, \texttt{omf} returns one $3$-cocycle per class and \texttt{omaut} one per orbit under the automorphisms of $G$, using also functions from \cite{HAP,mignard2017moritaequivalencepointedfusion}. The code
\begin{gapcode}
for name in [ "C2", "S3", "C2 x C2", "D8", "Q8" ] do
    G := MakeGroup(name);;
    Print(name, ": H^3 = ", H3Invariants(G), ", ", Length(omf(G)), " classes, ");
    Print(Length(omaut(G)), " up to automorphisms\n");
od;
\end{gapcode}
prints
\begin{gapout}
C2: H^3 = [ 2 ], 2 classes, 2 up to automorphisms
S3: H^3 = [ 6 ], 6 classes, 6 up to automorphisms
C2 x C2: H^3 = [ 2, 2, 2 ], 8 classes, 4 up to automorphisms
D8: H^3 = [ 2, 2, 4 ], 16 classes, 12 up to automorphisms
Q8: H^3 = [ 8 ], 8 classes, 8 up to automorphisms
\end{gapout}
A custom $3$-cocycle is a GAP record \texttt{rec(a,N)} with $\omega=e^{2\pi i\,a/N}$. We define the type-III cocycle of $\Z_2^3$
\be
\omega_{\rm III}(g_1,g_2,g_3)=(-1)^{x_1y_2z_3}
\ee
for group elements $g=a^xb^yc^z$, and check with \texttt{3coctest} that it is a normalized $3$-cocycle. 
\begin{gapcode}
K3 := MakeGroup("C2 x C2 x C2");;
a := GroupElement(K3, "a");; b := GroupElement(K3, "b");; c := GroupElement(K3, "c");;
xyz := g -> First(Cartesian([0,1], [0,1], [0,1]), e -> a^e[1]*b^e[2]*c^e[3] = g);;
omIII := rec(a := function(g1, g2, g3)
    return xyz(g1)[1]*xyz(g2)[2]*xyz(g3)[3] mod 2;
end, N := 2);;
3coctest([Elements(K3), omIII]);
\end{gapcode}
which returns
\begin{gapout}
true
\end{gapout}

An algebra $A(H,F,\gamma,\eps)$ can exist only if $\omega$ is trivial on $F$. \texttt{gammaf} returns the trivializations $\gamma$ (one per class of $H^2(F,U(1))$) or \texttt{fail}: we print the generators of the subgroups $F$, up to conjugacy, on which $\omega_{\rm III}$ is trivial, with their number of $\gamma$:
\begin{gapcode}
for F in List(ConjugacyClassesSubgroups(K3), Representative) do
    gs := gammaf([F, omIII]);;
    if gs <> fail then
        Print(List(GeneratorsOfGroup(F), g -> GroupLabel(K3, g)), ": ", Length(gs), "\n");
    fi;
od;
\end{gapcode}
where \texttt{<>} means ``not equal'' in GAP. The code prints, in three columns here,
\begin{gapout}
[  ]: 1          [ "c" ]: 1          [ "b", "c" ]: 2
[ "a" ]: 1       [ "ac" ]: 1         [ "a", "c" ]: 2
[ "b" ]: 1       [ "bc" ]: 1         [ "a", "b" ]: 2
[ "ab" ]: 1                          [ "ab", "ac" ]: 2
\end{gapout}
Every cyclic subgroup except $\langle abc\rangle$ admits one trivialization, four of the seven $\Z_2^2$ subgroups admit two, one per class of $H^2(\Z_2^2,U(1))=\Z_2$, and $\omega_{\rm III}$ is not trivial on $\Z_2^3$ itself. 

Finally, $D^{\omega_{\rm III}}(\Z_2^3)$ and $D(D_8)$ have the same number of anyons and the same topological spins, the diagonal of $T$, with their multiplicities:
\begin{gapcode}
thIII := TQDTheory(K3, omIII, rec(ST := true));;
thD8 := TQDTheory(MakeGroup("D8"), 0, rec(ST := true));;
Spins := th -> Collected(DiagonalOfMat(th.T));;
[ Length(thIII.aG), Length(thD8.aG), Spins(thIII) = Spins(thD8) ];
\end{gapcode}
returns \texttt{[ 22, 22, true ]}, as expected from the equivalence of the two theories, which we construct in Sec.~\ref{sec:ex4}. The script ends with \texttt{NormalizeCocycle}, which replaces a cocycle with $\omega(1,g,h)\neq1$ by a cohomologous normalized one, which is required by the algebra computation functions.

\subsection{Condensable Algebras}
\label{sec:ex4}
We have now seen how to set up and label groups, 3-cocycles and twisted quantum double theories. We can now compute the main objects of interest, condensable algebras with the example script 
\begin{gapcode}
Read("examples/4_algebras.g");
\end{gapcode}
\texttt{AlgebrasWithAnyons} returns every algebra $A(H,F,\gamma,\eps)$ as a record with its name, dimension, $K=H/F$ (the group for the reduced topological order) and anyon decomposition \eqref{eq:anyons}. \texttt{ShowAlgebrasWithAnyons} displays it by printing one line per algebra. For $D(S_3)$:
\begin{gapcode}
th := TQDTheory("S3");;
L := AlgebrasWithAnyons(th);;
ShowAlgebrasWithAnyons(L);
\end{gapcode}
prints
\begin{gapout}
1  dim=1  A(S3=<r,s>, 1, 1, 1)          K=S3  ([1],1)
2  dim=2  A(C3=<r>, 1, 1, 1)            K=C3  ([1],1)+([1],1 s=-1)
3  dim=3  A(S3=<r,s>, C3=<r>, 1, 1)     K=C2  ([1],1)+([r],1)
4  dim=3  A(C2=<s>, 1, 1, 1)            K=C2  ([1],1)+([1],2 r=-1)
5  dim=6  A(S3=<r,s>, S3=<r,s>, 1, 1)   K=1   ([1],1)+([s],1)+([r],1)
6  dim=6  A(C3=<r>, C3=<r>, 1, 1)       K=1   ([1],1)+([1],1 s=-1)+2([r],1)
7  dim=6  A(C2=<s>, C2=<s>, 1, 1)       K=1   ([1],1)+([1],2 r=-1)+([s],1)
8  dim=6  A(1, 1, 1, 1)                 K=1   ([1],1)+([1],1 s=-1)+2([1],2 r=-1)
\end{gapout}
The four Lagrangian algebras of dimension $|S_3|=6$ are the gapped boundaries, with $A(1,1,1,1)=\cL(1,1)$ the Dirichlet boundary for the $S_3$ symmetry and $A(S_3,S_3,1,1)=\cL(S_3,1)$ the Neumann boundary for $\Rep(S_3)$. 

\texttt{TQDAlgebras} can also filter the algebras, e.g. restrict to the Lagrangian ones with \texttt{TQDAlgebras(th, rec(Lag := true))} or those of a given quantum dimension. \texttt{FindAlgebras} finds the algebras that contain given anyons, and \texttt{ShowAlgData} prints one algebra with its reduced topological order. Finally, \texttt{AlgebraDisplayData} returns the same strings in \LaTeX{}, in the notation of~\cite{Gai:2026hjk}.

A 3-cocycle twist removes the algebras for which $\omega|_F$ is non-trivial in cohomology. For example, it is well-known that the untwisted $D(\Z_2\times\Z_2)$ has $16$ algebras and the type-II twisted double of Sec.~\ref{sec:ex1} only $7$. The long form of \texttt{ShowAlgData} also prints the non-trivial values of $\gamma$, $\eps$ and of the twist $\alpha$ of the reduced topological order $D^\alpha(H/F)$, computed by \texttt{alphaf}:
\begin{gapcode}
K := MakeGroup("C2 x C2");;
omII := First(omaut(K), om -> Length(GaugeableSubgroups(om)) = 3);;
thK := TQDTheory(K, omII);;
[ Length(TQDAlgebras(TQDTheory(K))), Length(TQDAlgebras(thK)) ];
ab := GroupElement(K, "a")*GroupElement(K, "b");;
r := First(AlgebrasWithAnyons(thK), r -> r.H = Group(ab) and IsTrivial(r.F));;
ShowAlgData(thK, r.A, rec(opP := "long"));;
\end{gapcode}
which prints
\begin{gapout}
[ 16, 7 ]
H=C2=<ab>   F=1  H/F=C2
alpha(f1*f2,f1*f2,f1*f2)=-1
  ([1],1)+([1],1 a=-1 b=-1)
\end{gapout}
which shows that condensing $A(\langle ab\rangle,1,1,1)$ produces the double semion theory $D^\alpha(\Z_2)$. Here, \texttt{f1*f2} id GAP's notation for the generator of the $\Z_2$ group of the reduced topological order.

\paragraf{The type-III twisted double and $D(D_8)$.} The script then builds $\Z_2^3=\langle R,G,B\rangle$ with the type-III cocycle of Sec.~\ref{sec:ex3} and names its $22$ anyons as in~\cite{Iqbal:2023wvm,Bhardwaj:2024qrf}: the electric charges $e_C$, the bosons and fermions with non-trivial flux $m_C$ and $f_C$, and the semions $s_{RGB}$, $\bar s_{RGB}$, where $C$ is a set of colors. \texttt{AlgebraClasses} sorts its $54$ algebra data into the $38$ inequivalent algebras of~\cite{Bhardwaj:2024qrf}, $11$ of them Lagrangian, since data related by the phases $\Omega_g$ of \eqref{eq:isomorphism} define isomorphic algebras \cite{Natale2017,Gai:2026hjk}. A Lagrangian algebra of $\cZ(\Vec_{\Z_2^3}^{\omega_{\rm III}})\boxtimes\ol{\cZ(\Vec_{D_8})}$ decomposes into a permutation of the anyons that preserves the modular $S$ and $T$ matrices, owing to the equivalence $\cZ(\Vec_{\Z_2^3}^{\omega_{\rm III}})\cong\cZ(\Vec_{D_8})$~\cite{deWildPropitius:1995cf,Iqbal:2023wvm,Coste:2000tq,Warman:2025hov}. The script computes and prints the $38$ algebras of $D(D_8)$ in these labels (the order of the lines and the subgroup representatives in the names can differ between GAP sessions),
\begin{gapout}
dim=1  A(D8=<r,s>, 1, 1, 1)  1
dim=2  A(D8=<r,s>, C2=<r^2>, 1, 1)  1+e_G
dim=2  A(D8=<r,s>, C2=<r^2>, 1, eps)  1+e_RG
dim=2  A(D8=<r,s>, C2=<r^2>, 1, eps)  1+e_GB
dim=2  A(D8=<r,s>, C2=<r^2>, 1, eps)  1+e_RGB
dim=2  A(C2 x C2=<rs,sr>, 1, 1, 1)  1+e_R
dim=2  A(C2 x C2=<s,r^2>, 1, 1, 1)  1+e_B
dim=2  A(C4=<r>, 1, 1, 1)  1+e_RB
dim=4  A(D8=<r,s>, C2 x C2=<rs,sr>, 1, 1)  1+m_B+e_G
dim=4  A(D8=<r,s>, C2 x C2=<rs,sr>, gamma, eps)  1+m_GB+e_RGB
dim=4  A(D8=<r,s>, C2 x C2=<s,r^2>, 1, 1)  1+m_R+e_G
dim=4  A(D8=<r,s>, C2 x C2=<s,r^2>, gamma, eps)  1+m_RG+e_RGB
dim=4  A(D8=<r,s>, C4=<r>, 1, 1)  1+e_G+m_RB
dim=4  A(D8=<r,s>, C4=<r>, 1, eps)  1+e_RGB+m_RB
dim=4  A(C2 x C2=<rs,sr>, C2=<rs>, 1, 1)  1+e_R+m_B
dim=4  A(C2 x C2=<rs,sr>, C2=<rs>, 1, eps)  1+e_R+m_GB
dim=4  A(C2 x C2=<rs,sr>, C2=<r^2>, 1, 1)  1+e_R+e_G+e_RG
dim=4  A(C2 x C2=<rs,sr>, C2=<r^2>, 1, eps)  1+e_R+e_RGB+e_GB
dim=4  A(C2 x C2=<s,r^2>, C2=<s>, 1, 1)  1+e_B+m_R
dim=4  A(C2 x C2=<s,r^2>, C2=<s>, 1, eps)  1+e_B+m_RG
dim=4  A(C2 x C2=<s,r^2>, C2=<r^2>, 1, 1)  1+e_B+e_G+e_GB
dim=4  A(C2 x C2=<s,r^2>, C2=<r^2>, 1, eps)  1+e_B+e_RG+e_RGB
dim=4  A(C4=<r>, C2=<r^2>, 1, 1)  1+e_RB+e_G+e_RGB
dim=4  A(C4=<r>, C2=<r^2>, 1, eps)  1+e_RB+e_RG+e_GB
dim=4  A(C2=<rs>, 1, 1, 1)  1+e_R+m_G
dim=4  A(C2=<s>, 1, 1, 1)  1+e_B+m_G
dim=4  A(C2=<r^2>, 1, 1, 1)  1+e_B+e_R+e_RB
dim=8  A(D8=<r,s>, D8=<r,s>, 1, 1)  1+m_B+m_R+e_G+m_RB
dim=8  A(D8=<r,s>, D8=<r,s>, gamma, eps)  1+m_GB+m_RG+e_RGB+m_RB
dim=8  A(C2 x C2=<rs,sr>, C2 x C2=<rs,sr>, 1, 1)  1+e_R+2m_B+e_G+e_RG
dim=8  A(C2 x C2=<rs,sr>, C2 x C2=<rs,sr>, gamma, eps)  1+e_R+2m_GB+e_RGB+e_GB
dim=8  A(C2 x C2=<s,r^2>, C2 x C2=<s,r^2>, 1, 1)  1+e_B+2m_R+e_G+e_GB
dim=8  A(C2 x C2=<s,r^2>, C2 x C2=<s,r^2>, gamma, eps)  1+e_B+2m_RG+e_RG+e_RGB
dim=8  A(C4=<r>, C4=<r>, 1, 1)  1+e_RB+e_G+e_RGB+2m_RB
dim=8  A(C2=<rs>, C2=<rs>, 1, 1)  1+e_R+m_G+m_B+m_GB
dim=8  A(C2=<s>, C2=<s>, 1, 1)  1+e_B+m_G+m_R+m_RG
dim=8  A(C2=<r^2>, C2=<r^2>, 1, 1)  1+e_B+e_R+e_RB+e_G+e_RG+e_RGB+e_GB
dim=8  A(1, 1, 1, 1)  1+e_B+e_R+e_RB+2m_G
The 38 algebras of D(D8) match the 38 classes of the type-III double
\end{gapout}
and checks that they have the anyons of the $38$ classes of $D^{\omega_{\rm III}}(\Z_2^3)$, which has no twin algebras. 

The number of algebra classes of every twisted double of $D_8$, with inequivalent $\om\in H^3(D_8,U(1))$ up to automorphisms, is computed by
\begin{gapcode}
D8 := MakeGroup("D8");;
SortedList(List(omaut(D8), om -> Length(AlgebraClasses(TQDTheory(D8, om)).classes)));
\end{gapcode}
which returns
\begin{gapout}
[ 8, 8, 9, 9, 10, 10, 11, 15, 19, 22, 27, 38 ]
\end{gapout}

\subsection{Interfaces and Maps of Anyons}
\label{sec:ex5}

The script
\begin{gapcode}
Read("examples/5_interface_map.g");
\end{gapcode}
computes interfaces between topological orders. Each algebra of the toric code is either the identity algebra or Lagrangian, so its only interfaces are the gapped boundaries and the invertible ones. \texttt{AutomorphismsAnyons} finds the invertible interfaces, as the Lagrangian algebras of $\cZ(\Vec_{\Z_2})\boxtimes\ol{\cZ(\Vec_{\Z_2})}$ whose decomposition is a permutation of the anyons:
\begin{gapcode}
thC2 := TQDTheory("C2");;
ShowAlgebrasWithAnyons(AlgebrasWithAnyons(thC2));
List(AutomorphismsAnyons(thC2), r -> thC2.labels{r.perm});
\end{gapcode}
which prints
\begin{gapout}
1  dim=1  A(C2=<a>, 1, 1, 1)        K=C2  1
2  dim=2  A(C2=<a>, C2=<a>, 1, 1)   K=1   1+m
3  dim=2  A(1, 1, 1, 1)             K=1   1+e
[ [ "1", "m", "e", "me" ], [ "1", "e", "m", "me" ] ]
\end{gapout}
The last line gives the images of the anyons $1$, $e$, $m$, $me$ under the electric-magnetic duality, which exchanges the two gapped boundaries $1+m$ and $1+e$, and under the identity. This is then checked with \texttt{AutomorphismsDG}, which uses the condition~\cite{Davydov2009ModularIF}, and also lists the anyons and algebras of the double semion. These are very simple examples, but they can easily be extended to larger quantum doubles $D(G)$.
\begin{table}
\centering
\scriptsize
\renewcommand{\arraystretch}{1.25}
\begin{tabular}{|c|c|}
\hline
$D(\Z_2\times\Z_2)$ & \nameref{A4:2} \\
\hline\hline
$e_1$, $e_2$, $e_1e_2$ & $([1],R_3)$ \\
$m_1$, $m_2$, $m_1m_2$ & $([b],1_{+,+})$ \\
$m_2e_1$, $m_1e_2$, $m_1m_2e_1e_2$ & $([b],1_{+,-})$ \\
$m_2e_2$, $m_1e_1e_2$, $m_1m_2e_1$ & $([b],1_{-,+})$ \\
$m_1e_1$, $m_2e_1e_2$, $m_1m_2e_2$ & $([b],1_{-,-})$ \\
\hline
\end{tabular}

\vspace{2ex}
\begin{tabular}{|c|c|c|c|}
\hline
$D(\Z_3)$ & \nameref{A4:3} & \nameref{A4:4} & \nameref{A4:5} \\
\hline\hline
$e$ & $([1],R_1)\oplus([b],1_{+,+})$ & $([1],R_1)\oplus([b],1_{+,-})$ & $([1],R_1)\oplus([1],R_3)$ \\
$e^2$ & $([1],R_2)\oplus([b],1_{+,+})$ & $([1],R_2)\oplus([b],1_{+,-})$ & $([1],R_2)\oplus([1],R_3)$ \\
\hline
$m$, $me$, $me^2$ & \multicolumn{3}{c|}{$([a],1)$, $([a],\rho)$, $([a],\rho^2)$} \\
$m^2$, $m^2e$, $m^2e^2$ & \multicolumn{3}{c|}{$([a^2],1)$, $([a^2],\rho^2)$, $([a^2],\rho)$} \\
\hline
\end{tabular}

\vspace{2ex}
\begin{tabular}{|c|c|c|c|}
\hline
$D(\Z_2)$ & \nameref{A4:6} & \nameref{A4:7} & \nameref{A4:8} \\
\hline\hline
$e$ & $([1],R_3)\oplus([b],1_{+,-})$ & $([1],R_3)\oplus([b],1_{+,+})$ & $2([1],R_3)$ \\
$m$ & $2([b],1_{+,+})$ & $2([b],1_{+,-})$ & $([b],1_{+,+})\oplus([b],1_{+,-})$ \\
\hline
$me$ & \multicolumn{3}{c|}{$([b],1_{-,+})\oplus([b],1_{-,-})$} \\
\hline
\end{tabular}
\caption{The canonical interface maps of the algebras of $D(A_4)$ in Table~\ref{tab:A4alg}, computed by \texttt{InterfaceMap} and written with the anyon notation of Table~\ref{tab:A4anyons}. Each column gives the image of the anyons of the reduced topological order across the SymTFT interface to anyons of $D(A_4)$.}
\label{tab:A4maps}
\end{table}
For a non-Lagrangian algebra, \texttt{InterfaceMap} computes the folded Lagrangian algebra $\cL^\folded(H^{\diag},\varphi)$ of \eqref{eq:map} and prints the anyons $b$ of the reduced topological order and their images $\bigoplus_an_{ab}\,a$ across the interface. It returns the transpose of the matrix $n$ of \eqref{eq:map}, with rows labeled by the reduced anyons and columns by the original ones. 
To illustrate its use, we take $D(A_4)$, which we will discuss in detail in Sec.~\ref{sec:ex6}, with the generators $a=(1\,2\,3)$ and $b=(1\,2)(3\,4)$. The anyons and algebras in $D(A_4)$ are listed in Tables~\ref{tab:A4anyons} and~\ref{tab:A4alg}. For example:
\begin{gapcode}
G := AlternatingGroup(4);;
MakeGroupLabels([G, [(1,2,3), (1,2)(3,4)], ["a","b"]]);;
thA4 := TQDTheory(G);;
M := AlgebraClasses(thA4);;
if Length(M.classes) <> 15 then
    ErrorNoReturn("D(A4) has ", Length(M.classes), " algebra classes instead of 15");
fi;
for c in Filtered(M.classes, c -> c.dim < Size(G)) do
    InterfaceMap(thA4, c.algebras[1]);;
od;
\end{gapcode}
returns the results summarized in Tab.~\ref{tab:A4maps}.

\subsection{Algebras, Hasse diagrams and Phases from $D^\om(A_4)$}
\label{sec:ex6}

The next example script
\begin{gapcode}
Read("examples/6_algebra_classes_and_inclusions.g");
\end{gapcode}
computes the isomorphism classes of algebras of $D^\om(A_4)$ with three choices of 3-cocycle $\omega$ (trivial, order-2 and order-3) and writes the Hasse diagrams of Figures~\ref{fig:HasseA4}, \ref{fig:HasseA4om2} and~\ref{fig:HasseA4om3} whose algebra labels are linked to the algebra tables. We present $A_4$ as:
\be\ba
    a=(1\,2\,3)\,,\qquad b=(1\,2)(3\,4)\,, \cr
    A_4=\langle a,b\mid a^3=b^2=1,\ b(aba^2)=(aba^2)b=a^2ba\rangle\,,
\ea\ee
with conjugacy classes $[1]$, $[a]=\{a,ab,ba,bab\}$, $[a^2]=\{a^2,a^2b,ba^2,aba\}$ and $[b]=\{b,aba^2,a^2ba\}$, centralizers $A_4$, $\Z_3=\langle a\rangle$, $\Z_3$ and $\Z_2\times\Z_2=\langle b,aba^2\rangle$, and irreducible representations $1$, $R_1$, $R_2$ with $R_{1,2}(a)=\zeta_3^{\pm1}$, and $R_3$ of dimension $3$.

The third cohomology group $H^3(A_4,U(1))=\Z_6$ has one class $\omega_2$ of order two, which we select by the order of its class, and a class $\omega_3$ of order three which can be pulled back from the quotient $\Z_3=A_4/(\Z_2\times\Z_2)=\langle m\rangle$, with its standard 3-cocycle \cite{Coste:2000tq} 
\be
\alpha(m^l,m^j,m^k)=\exp\bigl(\frac{2\pi i}{9}\,l(j+k-\langle j+k\rangle_3)\bigr)\,,\quad \langle j+k\rangle_3=(j+k)\bmod3 \,.
\ee
In the code, we use \texttt{gn(g)} for the exponent $l$ of $m^l\in \Z_3$ and define the pullback with:
\begin{gapcode}
G := AlternatingGroup(4);;
a := (1,2,3);; b := (1,2)(3,4);;
MakeGroupLabels([G, [a, b], ["a","b"]]);;
a := GroupElement(G, "a");; b := GroupElement(G, "b");;
n := H3Invariants(G)[1];;                   # H^3(A4,U(1)) = Z6
om2 := H3Cocycle(G, [ First([1..n-1], v -> n/Gcd(v, n) = 2) ]);;
p := NaturalHomomorphismByNormalSubgroup(G, Group(b, a*b*a^2));;   # A4 -> Z3
gn := g -> First([0..2], k -> Image(p, g) = Image(p, a)^k);;
om3 := rec(a := function(g1, g2, g3)
    return (gn(g1)*(gn(g2) + gn(g3) - ((gn(g2) + gn(g3)) mod 3))) mod 9;
end, N := 9);;
3coctest([Elements(G), om3]);
\end{gapcode}
which returns \texttt{true}. 

For each choice of $\omega$, \texttt{AlgebraClasses} sorts the algebras in $D^\om(A_4)$ into their classes and \texttt{Hasse} computes their subalgebra inclusions. Then, \texttt{PhaseLabels} names the phases relative to the Dirichlet boundary $\cL(1,1)=A(1,1,1,1)$, which realizes the $A_4$ symmetry with the anomaly $\omega$. The code
\begin{gapcode}
for om in [ 0, om2, om3 ] do
    th := TQDTheory(G, om, rec(ST := true));;
    M := AlgebraClasses(th);;  Hs := Hasse(th, M);;
    sym := PositionProperty(M.classes, c -> IsTrivial(c.algebras[1][1]));;
    Print(Length(M.classes), " algebra classes, L(1,1) is the algebra class ", sym, "\n");
    Print(PhaseLabels(th, sym, Hs), "\n");
od;
\end{gapcode}
prints
\begin{gapout}
15 algebra classes, L(1,1) is the algebra class 15
[ "gSPT", "gSSB", "gSPT", "gSPT", "gSSB", "gSSB", "gSSB", "gSSB", "SPT", "SPT", "SSB", "SSB", "SSB", "SSB", "SSB" ]
6 algebra classes, L(1,1) is the algebra class 6
[ "igSPT", "igSSB", "gSSB", "igSSB", "SSB", "SSB" ]
12 algebra classes, L(1,1) is the algebra class 12
[ "igSPT", "gSSB", "igSPT", "igSPT", "igSSB", "gSSB", "gSSB", "gSSB", "SSB", "SSB", "SSB", "SSB" ]
\end{gapout}
The phases are listed in the order of the classes, and the script writes each Hasse diagram with its linked algebra table as \LaTeX{} files with \texttt{HasseTableTeX}.

\begin{table}
\centering
\footnotesize
\renewcommand{\arraystretch}{0.9}
\begin{tabular}{|c|c|c|}
\hline
$D(A_4)$ & Dim & $T$ \\
\hline\hline
$([1],1)$ & $1$ & $1$ \\
$([1],R_1)$ & $1$ & $1$ \\
$([1],R_2)$ & $1$ & $1$ \\
$([1],R_3)$ & $3$ & $1$ \\
$([b],1_{+,+})$ & $3$ & $1$ \\
$([b],1_{+,-})$ & $3$ & $1$ \\
$([b],1_{-,+})$ & $3$ & $-1$ \\
$([b],1_{-,-})$ & $3$ & $-1$ \\
$([a],1)$ & $4$ & $1$ \\
$([a],\rho)$ & $4$ & $\zeta_3$ \\
$([a],\rho^2)$ & $4$ & $\zeta_3^2$ \\
$([a^2],1)$ & $4$ & $1$ \\
$([a^2],\rho)$ & $4$ & $\zeta_3$ \\
$([a^2],\rho^2)$ & $4$ & $\zeta_3^2$ \\
\hline
\end{tabular}
\qquad
\begin{tabular}{|c|c|c|}
\hline
$D^{\omega_2}(A_4)$ & Dim & $T$ \\
\hline\hline
$([1],1)$ & $1$ & $1$ \\
$([1],R_1)$ & $1$ & $1$ \\
$([1],R_2)$ & $1$ & $1$ \\
$([1],R_3)$ & $3$ & $1$ \\
$([b],1_{i,i})$ & $3$ & $i$ \\
$([b],1_{i,-i})$ & $3$ & $i$ \\
$([b],1_{-i,i})$ & $3$ & $-i$ \\
$([b],1_{-i,-i})$ & $3$ & $-i$ \\
$([a],1)$ & $4$ & $1$ \\
$([a],\rho)$ & $4$ & $\zeta_3$ \\
$([a],\rho^2)$ & $4$ & $\zeta_3^2$ \\
$([a^2],1)$ & $4$ & $1$ \\
$([a^2],\rho)$ & $4$ & $\zeta_3$ \\
$([a^2],\rho^2)$ & $4$ & $\zeta_3^2$ \\
\hline
\end{tabular}
\qquad
\begin{tabular}{|c|c|c|}
\hline
$D^{\omega_3}(A_4)$ & Dim & $T$ \\
\hline\hline
$([1],1)$ & $1$ & $1$ \\
$([1],R_1)$ & $1$ & $1$ \\
$([1],R_2)$ & $1$ & $1$ \\
$([1],R_3)$ & $3$ & $1$ \\
$([b],1_{+,+})$ & $3$ & $1$ \\
$([b],1_{+,-})$ & $3$ & $1$ \\
$([b],1_{-,+})$ & $3$ & $-1$ \\
$([b],1_{-,-})$ & $3$ & $-1$ \\
$([a],\zeta_9)$ & $4$ & $\zeta_9$ \\
$([a],\zeta_9^4)$ & $4$ & $\zeta_9^4$ \\
$([a],\zeta_9^7)$ & $4$ & $\zeta_9^7$ \\
$([a^2],\zeta_9)$ & $4$ & $\zeta_9$ \\
$([a^2],\zeta_9^4)$ & $4$ & $\zeta_9^4$ \\
$([a^2],\zeta_9^7)$ & $4$ & $\zeta_9^7$ \\
\hline
\end{tabular}
\caption{The anyons of $D(A_4)$, $D^{\omega_2}(A_4)$ and $D^{\omega_3}(A_4)$ as printed by \texttt{ShowAnyons}, with their quantum dimensions and topological spins $T$. The code labels a charge by the values of its character on the generators of the centralizer: $([b],1_{s_1,s_2})$ is the one-dimensional (projective) representation of $\Z_2^{\langle b\rangle}\times\Z_2^{\langle aba^2\rangle}$ with characters $s_1$ at $b$ and $s_2$ at $aba^2$, printed as \texttt{([b],1 aba\^{}2=-1)} for $1_{+,-}$, and $\rho$ is the representation of $\Z_3^{\langle a\rangle}$ with $\rho(a)=\zeta_3$. With $\omega_3$ the fluxes in $[a]$ and $[a^2]$ carry $\omega_3(\cdot,\cdot|a)^{-1}$-projective characters with values $\zeta_9^k$ and fractional spins, and with $\omega_2$ the characters of the fluxes in $[b]$ take the values $\pm i$.}
\label{tab:A4anyons}
\end{table}

\begin{figure}
\centering
\begin{tikzpicture}[vertex/.style={draw,font=\small}, xscale=1.2, yscale=0.9]
 \begin{scope}[shift={(0,0)}]
    \foreach \coord/\i/\j in {
(0.,-2)/1/{\nameref{A4:1} \\ gSPT},
(1.98,-2.7)/2/{\nameref{A4:2} \\ $\Z_3$\! gSSB},
(-4.2,-4.5)/3/{\nameref{A4:3} \\ gSPT},
(-2.9,-4.5)/4/{\nameref{A4:4} \\ gSPT},
(-1.4,-4.5)/5/{\nameref{A4:5} \\ $\Z_2^2$ gSSB},
(0.3,-4.5)/6/{\nameref{A4:6} \\ $\Z_3$\! gSSB},
(2,-4.5)/7/{\nameref{A4:7} \\ $\Z_3$\! gSSB},
(4.,-4.5)/8/{\nameref{A4:8} \\ $\!A_4/\Z_2\!$ gSSB},
(-4.3,-7)/9/{\nameref{A4:9} \\ SPT},
(-3.2,-7)/10/{\nameref{A4:10} \\ SPT${}^\gamma$},
(-1.9,-7)/11/{\nameref{A4:11} \\ $\Z_3$\! SSB},
(-0.5,-7)/12/{\nameref{A4:12} \\ $\Z_3$\! SSB},
(0.95,-7)/13/{\nameref{A4:13} \\ $\Z_2^2$ SSB},
(2.66,-7)/14/{\nameref{A4:14} \\ $\!A_4/\Z_2\!$ SSB},
(4.35,-7)/15/{\nameref{A4:15} \\ $\!A_4$ SSB}}
    {
      \node[vertex,align=center] (p\i) at \coord {\j};
     }
  \foreach [count=\r] \row in
{{0,1,1,1,1,0,0,0,0,0,0,0,0,0,0},{0,0,0,0,0,1,1,1,0,0,0,0,0,0,0},
{0,0,0,0,0,0,0,0,1,0,1,0,0,0,0},{0,0,0,0,0,0,0,0,0,1,0,1,0,0,0},
{0,0,0,0,0,0,0,0,0,0,0,0,1,0,1},{0,0,0,0,0,0,0,0,0,0,1,0,0,1,0},
{0,0,0,0,0,0,0,0,0,0,0,1,0,1,0},{0,0,0,0,0,0,0,0,0,0,0,0,0,1,1},
{0,0,0,0,0,0,0,0,0,0,0,0,0,0,0},{0,0,0,0,0,0,0,0,0,0,0,0,0,0,0},
{0,0,0,0,0,0,0,0,0,0,0,0,0,0,0},{0,0,0,0,0,0,0,0,0,0,0,0,0,0,0},
{0,0,0,0,0,0,0,0,0,0,0,0,0,0,0},{0,0,0,0,0,0,0,0,0,0,0,0,0,0,0},
{0,0,0,0,0,0,0,0,0,0,0,0,0,0,0}}
    {
     \foreach [count=\c] \cell in \row{
            \ifnum\cell=1%
                \draw[-stealth] (p\r) edge [thick] (p\c);
            \fi
        }
    }
\end{scope}
\end{tikzpicture}
\caption{Hasse diagram for $D(A_4)$, whose algebras are detailed in Table~\ref{tab:A4alg}. SPT$^\gamma$ denotes the non-trivial $A_4$-SPT.}
\label{fig:HasseA4}
\end{figure}
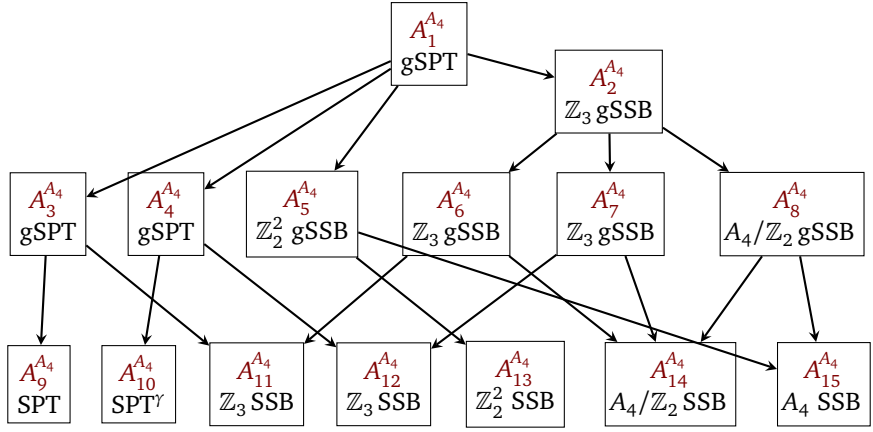

\begin{figure}
\centering
\begin{tikzpicture}[vertex/.style={draw,font=\small}, xscale=1.2, yscale=0.9]
 \begin{scope}[shift={(0,0)}]
    \foreach \coord/\i/\j in {
(0.,-1)/1/{\nameref{A4om2:1} \\igSPT},
(-2,-2.5)/2/{\nameref{A4om2:2} \\ $\Z_3$ igSSB},
(2,-4)/3/{\nameref{A4om2:3} \\ $\Z_2^2$ gSSB},
(-2,-4)/4/{\nameref{A4om2:4} \\ $\!A_4/\Z_2\!$ igSSB},
(2,-5.5)/5/{\nameref{A4om2:5} \\ $\Z_2^2$ SSB},
(-2,-5.5)/6/{\nameref{A4om2:6} \\ $A_4$ SSB}}
     {
      \node[vertex,align=center] (p\i) at \coord {\j};
     }
 \foreach [count=\r] \row in
{{0,1,1,0,0,0},
{0,0,0,1,0,0},
{0,0,0,0,1,1},
{0,0,0,0,0,1},
{0,0,0,0,0,0},
{0,0,0,0,0,0}}
    {
     \foreach [count=\c] \cell in \row{
            \ifnum\cell=1%
                \draw[-stealth] (p\r) edge [thick] (p\c);
            \fi
        }
    }
\end{scope}
\end{tikzpicture}
\caption{Hasse diagram for $D^{\omega_2}(A_4)$. The algebras are detailed in Table~\ref{tab:A4alg_omega2}.}
\label{fig:HasseA4om2}
\end{figure}

\begin{figure}
\centering
\begin{tikzpicture}[vertex/.style={draw,font=\small}, xscale=1.2, yscale=0.9]
 \begin{scope}[shift={(0,0)}]
    \foreach \coord/\i/\j in {
(0.,-2)/1/{\nameref{A4om3:1} \\ igSPT},
(1.98,-3)/2/{\nameref{A4om3:2} \\ $\Z_3$\! gSSB},
(-4.2,-4.5)/3/{\nameref{A4om3:3} \\ igSPT},
(-2.8,-4.5)/4/{\nameref{A4om3:4} \\ igSPT},
(2.125,-4.5)/5/{\nameref{A4om3:5} \\ $\Z_2^2$ igSSB},
(-1.22,-4.5)/6/{\nameref{A4om3:6} \\ $\Z_3$\! gSSB},
(0.5,-4.5)/7/{\nameref{A4om3:7} \\ $\Z_3$\! gSSB},
(4,-4.5)/8/{\nameref{A4om3:8} \\ $\!A_4/\Z_2\!$ gSSB},
(-4.2,-7)/9/{\nameref{A4om3:9} \\ $\Z_3$\! SSB},
(-1.5,-7)/10/{\nameref{A4om3:10} \\ $\Z_3$\! SSB},
(1.5,-7)/11/{\nameref{A4om3:11} \\ $\!A_4/\Z_2\!$ SSB},
(4,-7)/12/{\nameref{A4om3:12} \\ $A_4$ SSB}}
    {
      \node[vertex,align=center] (p\i) at \coord {\j};
     }
  \foreach [count=\r] \row in
  {{0,1,1,1,1,0,0,0,0,0,0,0},
{0,0,0,0,0,1,1,1,0,0,0,0},
{0,0,0,0,0,0,0,0,1,0,0,0},
{0,0,0,0,0,0,0,0,0,1,0,0},
{0,0,0,0,0,0,0,0,0,0,0,1},
{0,0,0,0,0,0,0,0,1,0,1,0},
{0,0,0,0,0,0,0,0,0,1,1,0},
{0,0,0,0,0,0,0,0,0,0,1,1},
{0,0,0,0,0,0,0,0,0,0,0,0},
{0,0,0,0,0,0,0,0,0,0,0,0},
{0,0,0,0,0,0,0,0,0,0,0,0},
{0,0,0,0,0,0,0,0,0,0,0,0}}
    {
     \foreach [count=\c] \cell in \row{
            \ifnum\cell=1%
                \draw[-stealth] (p\r) edge [thick] (p\c);
            \fi
        }
    }
\end{scope}
\end{tikzpicture}
\caption{Hasse diagram for $D^{\omega_3}(A_4)$, whose algebras are detailed in Table~\ref{tab:A4alg_omega3}.}
\label{fig:HasseA4om3}
\end{figure}

\begin{table}
\centering
\scriptsize
\renewcommand{\arraystretch}{1.24}
\microtypesetup{protrusion=true,expansion=true}
\begin{minipage}[t]{0.48\textwidth}
\vspace{0pt}
\centering
\captionsetup{width=\linewidth}
\begin{tabular}{|c|}
\hline
Condensable algebras of $D(A_4)$ \\
\hline\hline
$\phantom{D(A_4)}$\quad\hfill\nameref{A4:1}$=A(A_4, 1, 1, 1)$\hfill\quad$D(A_4)$\\
$\xlabel[$A^{A_4}_{1}$]{A4:1}([1],1)$\\
\hline
$\phantom{D(\Z_2\times\Z_2)}$\quad\hfill\nameref{A4:2}$=A(\Z_2\times\Z_2, 1, 1, 1)$\hfill\quad$D(\Z_2\times\Z_2)$\\
$\xlabel[$A^{A_4}_{2}$]{A4:2}([1],1) \oplus ([1],R_1) \oplus ([1],R_2)$\\
\hline
$\phantom{D(\Z_3)}$\quad\hfill\nameref{A4:3}$=A(A_4, \Z_2\times\Z_2, 1, 1)$\hfill\quad$D(\Z_3)$\\
$\xlabel[$A^{A_4}_{3}$]{A4:3}([1],1) \oplus ([b],1_{+,+})$\\
\hline
$\phantom{D(\Z_3)}$\quad\hfill\nameref{A4:4}$=A(A_4, \Z_2\times\Z_2, \gamma, \eps)$\hfill\quad$D(\Z_3)$\\
$\xlabel[$A^{A_4}_{4}$]{A4:4}([1],1) \oplus ([b],1_{+,-})$\\
\hline
$\phantom{D(\Z_3)}$\quad\hfill\nameref{A4:5}$=A(\Z_3, 1, 1, 1)$\hfill\quad$D(\Z_3)$\\
$\xlabel[$A^{A_4}_{5}$]{A4:5}([1],1) \oplus ([1],R_3)$\\
\hline
$\phantom{D(\Z_2)}$\quad\hfill\nameref{A4:6}$=A(\Z_2\times\Z_2, \Z_2^{\langle aba^2\rangle}, 1, 1)$\hfill\quad$D(\Z_2)$\\
$\xlabel[$A^{A_4}_{6}$]{A4:6}([1],1) \oplus ([1],R_1) \oplus ([1],R_2) \oplus ([b],1_{+,+})$\\
\hline
$\phantom{D(\Z_2)}$\quad\hfill\nameref{A4:7}$=A(\Z_2\times\Z_2, \Z_2^{\langle aba^2\rangle}, 1, \eps)$\hfill\quad$D(\Z_2)$\\
$\xlabel[$A^{A_4}_{7}$]{A4:7}([1],1) \oplus ([1],R_1) \oplus ([1],R_2) \oplus ([b],1_{+,-})$\\
\hline
$\phantom{D(\Z_2)}$\quad\hfill\nameref{A4:8}$=A(\Z_2, 1, 1, 1)$\hfill\quad$D(\Z_2)$\\
$\xlabel[$A^{A_4}_{8}$]{A4:8}([1],1) \oplus ([1],R_1) \oplus ([1],R_2) \oplus ([1],R_3)$\\
\hline\hline
\nameref{A4:9}$=A(A_4, A_4, 1, 1)$\\
$\xlabel[$A^{A_4}_{9}$]{A4:9}([1],1) \oplus ([b],1_{+,+}) \oplus ([a],1) \oplus ([a^2],1)$\\
\hline
\nameref{A4:10}$=A(A_4, A_4, \gamma, \eps)$\\
$\xlabel[$A^{A_4}_{10}$]{A4:10}([1],1) \oplus ([b],1_{+,-}) \oplus ([a],1) \oplus ([a^2],1)$\\
\hline
\nameref{A4:11}$=A(\Z_2\times\Z_2, \Z_2\times\Z_2, 1, 1)$\\
$\xlabel[$A^{A_4}_{11}$]{A4:11}([1],1) \oplus ([1],R_1) \oplus ([1],R_2) \oplus 3([b],1_{+,+})$\\
\hline
\nameref{A4:12}$=A(\Z_2\times\Z_2, \Z_2\times\Z_2, \gamma, \eps)$\\
$\xlabel[$A^{A_4}_{12}$]{A4:12}([1],1) \oplus ([1],R_1) \oplus ([1],R_2) \oplus 3([b],1_{+,-})$\\
\hline
\nameref{A4:13}$=A(\Z_3, \Z_3, 1, 1)$\\
$\xlabel[$A^{A_4}_{13}$]{A4:13}([1],1) \oplus ([1],R_3) \oplus ([a],1) \oplus ([a^2],1)$\\
\hline
\nameref{A4:14}$=A(\Z_2, \Z_2, 1, 1)$\\
$\xlabel[$A^{A_4}_{14}$]{A4:14}([1],1) \oplus ([1],R_1) \oplus ([1],R_2) \oplus ([1],R_3)$\\
$\oplus ([b],1_{+,+}) \oplus ([b],1_{+,-})$\\
\hline
\nameref{A4:15}$=A(1, 1, 1, 1)$\\
$\xlabel[$A^{A_4}_{15}$]{A4:15}([1],1) \oplus ([1],R_1) \oplus ([1],R_2) \oplus 3([1],R_3)$\\
\hline
\end{tabular}
\caption{The condensable algebras of $D(A_4)$, in the order of the output of \texttt{AlgebraClasses}. Each entry gives the data $A(H,F,\gamma,\eps)$ with the reduced topological order on its right, and below it the decomposition into the anyons of Table~\ref{tab:A4anyons}. The subgroups are $\Z_2\times\Z_2=\langle b,aba^2\rangle$, $\Z_3=\langle ab\rangle$ and $\Z_2=\langle b\rangle$ unless stated otherwise, and $\gamma$ is the non-trivial class of $H^2(F,U(1))=\Z_2$. The Lagrangian algebras follow the double line, and their reduced topological order is trivial.}
\label{tab:A4alg}
\end{minipage}
\hfill
\begin{minipage}[t]{0.48\textwidth}
\vspace{0pt}
\centering
\captionsetup{width=\linewidth}
\begin{tabular}{|c|}
\hline
Condensable algebras of $D^{\omega_2}(A_4)$ \\
\hline\hline
$\phantom{D^{\omega_2}(A_4)}$\quad\hfill\nameref{A4om2:1}$=A(A_4, 1, 1, 1)$\hfill\quad$D^{\omega_2}(A_4)$\\
$\xlabel[$A^{A_4,\omega_2}_{1}$]{A4om2:1}([1],1)$\\
\hline
$\phantom{D^{\alpha}(\Z_2\times\Z_2)}$\quad\hfill\nameref{A4om2:2}$=A(\Z_2\times\Z_2, 1, 1, 1)$\hfill\quad$D^{\alpha}(\Z_2\times\Z_2)$\\
$\xlabel[$A^{A_4,\omega_2}_{2}$]{A4om2:2}([1],1) \oplus ([1],R_1) \oplus ([1],R_2)$\\
\hline
$\phantom{D(\Z_3)}$\quad\hfill\nameref{A4om2:3}$=A(\Z_3, 1, 1, 1)$\hfill\quad$D(\Z_3)$\\
$\xlabel[$A^{A_4,\omega_2}_{3}$]{A4om2:3}([1],1) \oplus ([1],R_3)$\\
\hline
$\phantom{D^{\alpha}(\Z_2)}$\quad\hfill\nameref{A4om2:4}$=A(\Z_2, 1, 1, 1)$\hfill\quad$D^{\alpha}(\Z_2)$\\
$\xlabel[$A^{A_4,\omega_2}_{4}$]{A4om2:4}([1],1) \oplus ([1],R_1) \oplus ([1],R_2) \oplus ([1],R_3)$\\
\hline\hline
\nameref{A4om2:5}$=A(\Z_3, \Z_3, 1, 1)$\\
$\xlabel[$A^{A_4,\omega_2}_{5}$]{A4om2:5}([1],1) \oplus ([1],R_3) \oplus ([a],1) \oplus ([a^2],1)$\\
\hline
\nameref{A4om2:6}$=A(1, 1, 1, 1)$\\
$\xlabel[$A^{A_4,\omega_2}_{6}$]{A4om2:6}([1],1) \oplus ([1],R_1) \oplus ([1],R_2) \oplus 3([1],R_3)$\\
\hline
\end{tabular}
\caption{The condensable algebras of $D^{\omega_2}(A_4)$, with the conventions of Table~\ref{tab:A4alg}. The reduced topological orders of the second and fourth algebra are twisted, with $\alpha$ the 3-cocycle of the reduced topological order computed by \texttt{alphaf}.}
\label{tab:A4alg_omega2}
\vspace{3ex}
\begin{tabular}{|c|}
\hline
Condensable algebras of $D^{\omega_3}(A_4)$ \\
\hline\hline
$\phantom{D^{\omega_3}(A_4)}$\quad\hfill\nameref{A4om3:1}$=A(A_4, 1, 1, 1)$\hfill\quad$D^{\omega_3}(A_4)$\\
$\xlabel[$A^{A_4,\omega_3}_{1}$]{A4om3:1}([1],1)$\\
\hline
$\phantom{D(\Z_2\times\Z_2)}$\quad\hfill\nameref{A4om3:2}$=A(\Z_2\times\Z_2, 1, 1, 1)$\hfill\quad$D(\Z_2\times\Z_2)$\\
$\xlabel[$A^{A_4,\omega_3}_{2}$]{A4om3:2}([1],1) \oplus ([1],R_1) \oplus ([1],R_2)$\\
\hline
$\phantom{D^{\alpha}(\Z_3)}$\quad\hfill\nameref{A4om3:3}$=A(A_4, \Z_2\times\Z_2, 1, 1)$\hfill\quad$D^{\alpha}(\Z_3)$\\
$\xlabel[$A^{A_4,\omega_3}_{3}$]{A4om3:3}([1],1) \oplus ([b],1_{+,+})$\\
\hline
$\phantom{D^{\alpha}(\Z_3)}$\quad\hfill\nameref{A4om3:4}$=A(A_4, \Z_2\times\Z_2, \gamma, \eps)$\hfill\quad$D^{\alpha}(\Z_3)$\\
$\xlabel[$A^{A_4,\omega_3}_{4}$]{A4om3:4}([1],1) \oplus ([b],1_{+,-})$\\
\hline
$\phantom{D^{\alpha}(\Z_3)}$\quad\hfill\nameref{A4om3:5}$=A(\Z_3, 1, 1, 1)$\hfill\quad$D^{\alpha}(\Z_3)$\\
$\xlabel[$A^{A_4,\omega_3}_{5}$]{A4om3:5}([1],1) \oplus ([1],R_3)$\\
\hline
$\phantom{D(\Z_2)}$\quad\hfill\nameref{A4om3:6}$=A(\Z_2\times\Z_2, \Z_2^{\langle aba^2\rangle}, 1, 1)$\hfill\quad$D(\Z_2)$\\
$\xlabel[$A^{A_4,\omega_3}_{6}$]{A4om3:6}([1],1) \oplus ([1],R_1) \oplus ([1],R_2) \oplus ([b],1_{+,+})$\\
\hline
$\phantom{D(\Z_2)}$\quad\hfill\nameref{A4om3:7}$=A(\Z_2\times\Z_2, \Z_2^{\langle aba^2\rangle}, 1, \eps)$\hfill\quad$D(\Z_2)$\\
$\xlabel[$A^{A_4,\omega_3}_{7}$]{A4om3:7}([1],1) \oplus ([1],R_1) \oplus ([1],R_2) \oplus ([b],1_{+,-})$\\
\hline
$\phantom{D(\Z_2)}$\quad\hfill\nameref{A4om3:8}$=A(\Z_2, 1, 1, 1)$\hfill\quad$D(\Z_2)$\\
$\xlabel[$A^{A_4,\omega_3}_{8}$]{A4om3:8}([1],1) \oplus ([1],R_1) \oplus ([1],R_2) \oplus ([1],R_3)$\\
\hline\hline
\nameref{A4om3:9}$=A(\Z_2\times\Z_2, \Z_2\times\Z_2, 1, 1)$\\
$\xlabel[$A^{A_4,\omega_3}_{9}$]{A4om3:9}([1],1) \oplus ([1],R_1) \oplus ([1],R_2) \oplus 3([b],1_{+,+})$\\
\hline
\nameref{A4om3:10}$=A(\Z_2\times\Z_2, \Z_2\times\Z_2, \gamma, \eps)$\\
$\xlabel[$A^{A_4,\omega_3}_{10}$]{A4om3:10}([1],1) \oplus ([1],R_1) \oplus ([1],R_2) \oplus 3([b],1_{+,-})$\\
\hline
\nameref{A4om3:11}$=A(\Z_2, \Z_2, 1, 1)$\\
$\xlabel[$A^{A_4,\omega_3}_{11}$]{A4om3:11}([1],1) \oplus ([1],R_1) \oplus ([1],R_2) \oplus ([1],R_3)$\\
$\oplus ([b],1_{+,+}) \oplus ([b],1_{+,-})$\\
\hline
\nameref{A4om3:12}$=A(1, 1, 1, 1)$\\
$\xlabel[$A^{A_4,\omega_3}_{12}$]{A4om3:12}([1],1) \oplus ([1],R_1) \oplus ([1],R_2) \oplus 3([1],R_3)$\\
\hline
\end{tabular}
\caption{The condensable algebras of $D^{\omega_3}(A_4)$, with the conventions of Table~\ref{tab:A4alg}. The three algebras with reduced topological order $D^{\alpha}(\Z_3)$, with $\alpha$ non-trivial, are intrinsically gapless in Figure~\ref{fig:HasseA4om3}.}
\label{tab:A4alg_omega3}
\end{minipage}
\end{table}

The anyons of the three theories are listed in Table~\ref{tab:A4anyons}, in the labels of the output, and the algebras in Tables~\ref{tab:A4alg}, \ref{tab:A4alg_omega2} and \ref{tab:A4alg_omega3}, whose entries are the nodes of the Hasse diagrams. The symmetry boundary is the Dirichlet algebra $\cL(1,1)$, which realizes the $A_4$ symmetry with the anomaly $\omega$, and the phase types in the diagrams are the output of \texttt{PhaseLabels}, with the preserved subgroup of an SSB phase given by $H$. 

Without anomaly, $A_4$ has fifteen phases, which agree with~\cite{Aksoy:2025rmg}, of which seven are gapped: the trivial phase $\cL(A_4,1)$ and the $A_4$ SPT phase $\cL(A_4,\gamma)$ with $\gamma$ the generator of $H^2(A_4,U(1))=\Z_2$, studied numerically in~\cite{Lootens:2024gfp}, which differ only in whether the flux $([b],1_{+,+})$ or $([b],1_{+,-})$ is condensed, and five SSB phases. We denote by $A_4/\Z_2$ SSB the phase preserving $H=\Z_2^{\langle b\rangle}=\langle b\rangle$.

We remark that the algebra $A^{A_4}_5$ is not a subalgebra of $A^{A_4}_{14}$, although its anyons are a subset of those of $A^{A_4}_{14}$, and the same happens for $A^{A_4,\omega_3}_5$ and $A^{A_4,\omega_3}_{11}$. Indeed, the inclusion test of \texttt{IsSubalgebra} is the partial order of \cite{Gai:2026hjk} summarized in Sec.~\ref{sec:algebras} and not merely the subobject condition (which would be incorrect), a phenomenon also found for the Haagerup symmetry~\cite{Bottini:2025hri}.

For the $A_4$ symmetry with the $3$-cocycle anomaly $\omega_2$, which is non-trivial on $\Z_2\times\Z_2$ and on $\Z_2=\langle b\rangle$, only $\Z_3$ and the trivial group admit a gapped phase. The gapless phases whose reduced topological order has a non-trivial twist, $A^{A_4,\omega_2}_1$, $A^{A_4,\omega_2}_2$ and $A^{A_4,\omega_2}_4$, are intrinsically gapless. For the $A_4$ symmetry with the $3$-cocycle anomaly $\omega_3$, which is non-trivial on $\Z_3$ only, the gapped phases are those with $H\leq\Z_2\times\Z_2$, while the three gapless phases with reduced topological order $D^{\alpha}(\Z_3)$, as well as $A^{A_4,\omega_3}_1$, are intrinsically gapless.

\subsection{Twin Algebras of \texorpdfstring{$G_{32}$}{G32}}
\label{sec:ex7}

The computational core of~\cite{Gai:2026hjk} is the code
\begin{gapcode}
Read("examples/7_twin_algebras_G32.g");
\end{gapcode}
No group of order $<32$ has twin algebras (inequivalent algebras with the same anyon decomposition) in $D(G)$, as the full sweep of the script verifies (see below), and the first twin algebras appear for certain groups of order 32. We will consider a group which was studied in \cite{Gai:2026hjk} using this code and denote it by 
\be
G_{32}=\Z_8\rtimes(\Z_2\times\Z_2)=\texttt{SmallGroup(32,43)}
\ee
We define the generators in the notation of~\cite{Gai:2026hjk}: 
\be\ba
G_{32}=\langle g_1,g_2,g_3\mid g_1^8=g_2^2=g_3^2=\id\,,\ g_2g_3=g_3g_2\,,\ g_2g_1g_2^{-1}=g_1^3\,,\ g_3g_1g_3^{-1}=g_1^5\rangle \,.
\ea\ee
We label the group generators as follows:
\begin{gapcode}
G := MakeGroup([32,43]);;
gens := First(Cartesian(Elements(G), Elements(G), Elements(G)), t ->
    Order(t[1]) = 8 and Order(t[2]) = 2 and Order(t[3]) = 2 and t[2]*t[3] = t[3]*t[2]
    and t[2]*t[1]/t[2] = t[1]^3 and t[3]*t[1]/t[3] = t[1]^5 and Group(t) = G);;
RelabelGroup(G, ["g2","g3","g1"], [gens[2], gens[3], gens[1]]);;
\end{gapcode}
\texttt{AlgebraClasses} returns the algebra classes, the twin pairs of non-isomorphic algebras with the same anyons, and for each pair the reason why the two algebras are not isomorphic, as a \LaTeX{} string:
\begin{gapcode}
th := TQDTheory(G);;
M := AlgebraClasses(th);;
Length(M.classes);
for k in [1..Length(M.twins)] do
    Print(M.reasons[k][2], "\n");
    for i in M.twins[k] do Print(" ", M.classes[i].name, "\n"); od;
    Print(" = ", M.classes[M.twins[k][1]].decomposition, "\n");
od;
\end{gapcode}
which prints
\begin{gapout}
517
\eps^{\prime g} \neq \eps
  A(C8 : (C2 x C2)=<g2,g3,g1>, C8=<g1>, 1, 1)
  A(C8 : (C2 x C2)=<g2,g3,g1>, C8=<g1>, 1, eps)
  = ([1],1)+([g1^2],1)+([g1^4],1)+([g1],1)
\eps^{\prime g} \neq \eps
  A(C8 : (C2 x C2)=<g2,g3,g1>, C8=<g3g1>, 1, 1)
  A(C8 : (C2 x C2)=<g2,g3,g1>, C8=<g3g1>, 1, eps)
  = ([1],1)+([g1^2],1)+([g1^4],1)+([g3g1],1)
(H',F') \neq {}^g (H,F)
  A(C2 x C2 x C2=<g2,g3,g1^4>, C2=<g3>, 1, eps)
  A(C2 x C2 x C2=<g2,g3,g1^4>, C2=<g3g1^4>, 1, eps)
  = ([1],1)+([1],1 g1=-1)+([1],2)+([g3],2)
H' \neq {}^g H
  A(C2 x C2=<g2,g3>, 1, 1, 1)
  A(C2 x C2=<g2g3,g2g1^4>, 1, 1, 1)
  = ([1],1)+([1],1 g1=-1)+([1],2)+([1],4)
\end{gapout}
The $556$ algebra data fall into $517$ classes with four twin pairs, all non-Lagrangian: two $\eps$-type pairs, which are the gapless SPT twins distinguished by generalized string order parameters, one $F$-type pair (called $N$-type in~\cite{Gai:2026hjk}), and one $H$-type pair from the Gassmann triple $(G_{32},H_1,H_2)$ of two non-conjugate $\Z_2^2$ groups~\cite{Gai:2026hjk}. The normal subgroups $\langle g_3\rangle$ and $\langle g_3g_1^4\rangle={}^{g_1}\langle g_3\rangle$ of the $F$-type pair are conjugate in $G$, but not by the normalizer of $H_3=\langle g_2,g_3,g_1^4\rangle$:
\begin{gapcode}
g1 := GroupElement(G, "g1");; g2 := GroupElement(G, "g2");; g3 := GroupElement(G, "g3");;
H3 := Group(g2, g3, g1^4);;  F1 := Group(g3);;  F2 := Group(g3*g1^4);;
[ IsConjugate(G, F1, F2), IsConjugate(Normalizer(G, H3), F1, F2) ];
\end{gapcode}
returns \texttt{[ true, false ]}. For the $H$-type pair, the only one with $F=1$, the inclusions \texttt{[i,j]} of class $i$ into class $j$ from the \texttt{Hasse} function give the largest common subalgebra of the twins:
\begin{gapcode}
Hs := Hasse(th, M);;
pair := First(M.twins, p -> ForAll(p, i -> IsTrivial(M.classes[i].algebras[1][2])));;
common := Filtered([1..Length(M.classes)],
    i -> IsSubset(Hs.inclusions, [[i, pair[1]], [i, pair[2]]]));;
Maximum(List(common, i -> [M.classes[i].dim, M.classes[i].name]));
\end{gapcode}
which returns
\begin{gapout}
[ 4, "A(C2 x C2 x C2=<g2,g3,g1^4>, 1, 1, 1)" ]
\end{gapout}
which is the interface through which those two $H$-type twins are reached. 

The script then asks the user two yes/no questions. The first is whether to compute the algebras with the $3$-cocycle $\omega$ of \texttt{examples/\allowbreak G32-43-omega.g}, which restricts to the type-III cocycle on $H_3=\Z_2^3$. The resulting theory has $5$ twin pairs, all of $H$-type, including the twin Lagrangian algebras $\cL(H_1,1)$ and $\cL(H_2,1)$ that define the twin gapped phases of~\cite{Gai:2026hjk}. The $F$-type pair does not survive the twist, since $\omega(\cdot,\cdot|g_3)$ is not a coboundary on $H_3$. The second question is whether to run the sweep that reproduces the table of twins in $D(G)$ of~\cite{Gai:2026hjk}, by default over the groups listed there and optionally over all groups of order up to $48$ with \texttt{fullSweep := true}.

\subsection{Twin Phases of GL(2,3)}
\label{sec:ex8}

The next script, to be started from a clean GAP session,
\begin{gapcode}
Read("examples/8_twin_algebras_GL23.g");
\end{gapcode}
was used for the twin phases of~\cite{Warman:2026gfz}. The group $G=GL(2,3)\cong\allowbreak Q_8\rtimes S_3$ of order $48$ is constructed as a subgroup of the single-qubit Clifford group. We build it from the Hadamard gate \texttt{Had} and the phase gate \texttt{Phase}, with the generators $h$, $c$, $a$, $xz$, $iz$ of~\cite{Warman:2026gfz}:
\begin{gapcode}
Had := [[1,1],[1,-1]]/Sqrt(2);;  Phase := DiagonalMat([1,E(4)]);;
x := [[0,1],[1,0]];;  z := DiagonalMat([1,-1]);;
iso := IsomorphismPcGroup(Group([Had, (Had*Phase)^8]));;
G := MakeGroup(Image(iso));;
h := Image(iso, Had);;  c := Image(iso, -IdentityMat(2));;
a := Image(iso, E(8)^3*Had*Phase*x);;
RelabelGroup(G, ["h","c","a","xz","iz"], [h, c, a, Image(iso, x*z), Image(iso, E(4)*z)]);;
\end{gapcode}
The subgroups $S_3^{(1)}=\langle a,h\rangle$ and $S_3^{(2)}=\langle a,hc\rangle$ are not conjugate but form a Gassmann triple, and $D(G)$ has $86$ algebra classes with $5$ twin pairs:
\begin{gapcode}
S31 := Group(a, h);; S32 := Group(a, h*c);; D12 := Group(a, h, c);; Z3 := Group(a);;
th := TQDTheory(G);;
M := AlgebraClasses(th);;
[ Length(M.classes), Length(M.twins) ];
\end{gapcode}
which returns \texttt{[86,5]}. 

\texttt{ClassOn(H,F)} returns the position of the class of the algebras on the subgroups $H$ and $F$. The twin Lagrangian algebras $\cL(S_3^{(1)},1)$ and $\cL(S_3^{(2)},1)$ have the same anyons, and, as in Sec.~\ref{sec:ex7}, the inclusions of \texttt{Hasse} show that $A(D_{12},\Z_3,1,1)$ is a largest common subalgebra of the two twins:
\begin{gapcode}
ClassOn := {H, F} -> PositionProperty(M.classes,
    cl -> ForAny(cl.algebras, A -> ForAny(G, g -> A[1]^g = H and A[2]^g = F)));;
k := [ ClassOn(S31, S31), ClassOn(S32, S32), ClassOn(D12, Z3) ];;
for i in k do Print(M.classes[i].name, "\n"); od;
M.classes[k[1]].decomposition = M.classes[k[2]].decomposition;
Hs := Hasse(th, M);;
common := Filtered([1..Length(M.classes)],
    i -> IsSubset(Hs.inclusions, [[i, k[1]], [i, k[2]]]));;
[ Maximum(List(common, i -> M.classes[i].dim)), k[3] in common ];
\end{gapcode}
which prints
\begin{gapout}
A(S3=<h,a>, S3=<h,a>, 1, 1)
A(S3=<a,hc>, S3=<a,hc>, 1, 1)
A(D12=<hc,ha>, C3=<a>, 1, 1)
true
[ 12, true ]
\end{gapout}
where the last line shows that the largest common subalgebras have dimension $12$ and that $A(D_{12},\Z_3,1,1)$ is one of them.
In the labels of~\cite{Warman:2026gfz}, used by the script, the twins decompose as 
\be
1+\rho_3+\rho_4+[h]_{++}+[h]_{+-}+[a]_{++}+[a]_{+-}
\ee
and $A(D_{12},\Z_3,1,1)$ as $1+\rho_3+[a]_{++}$. The twin Lagrangian algebras give rise to the twin gapped phases, whose order parameters are different components of $\rho_4$.
$A(D_{12},\Z_3,1,1)$ specifies an interface that can be used for the transition, from the ``club sandwich'' SymTFT construction~\cite{Bhardwaj:2023bbf,Bhardwaj:2024qrf,Chatterjee:2022tyg,Warman:2026gfz}.

The code then asks if one would like to compute the twin algebras with the anomaly $\omega$ of \texttt{examples/\allowbreak GL23-om-data.g} which is trivial on both $S_3^{(k)}$ and pulls back the type-II cocycle on $D_{12}/\Z_3=\Z_2\times\Z_2$. This implies that no Lagrangian algebra has $H=D_{12}$ so the transition between twins cannot be written as a composition of two Landau transitions. The code computes the map of anyons across $A(D_{12},\Z_3,1,1)$ to its reduced topological order $D^{\omega_{\rm II}}(\Z_2\times\Z_2)$,
\begin{gapout}
([1],1)            ->  1+rho_3+[a]_{++}
([1],1 h=-1)       ->  rho_{1-}+rho_{3-}+[a]_{++}
([1],1 c=-1 h=-1)  ->  rho_4+[a]_{+-}
([hc],1 c=E(4)^3)  ->  [h]_{++}
([h],1)            ->  [h]_{++}          ...
All 16 anyon-map rows match (InterfaceMap), including [c]_{-i} times the G-irreps
and the [ac]_{+-i}, [ac]_{++i} terms.
\end{gapout}
Finally, \texttt{PhaseLabels} confirms that the SSB and trivial $\Z_4$ phases of the reduced topological order map to the two twin full SSB phases of the non-invertible symmetry $\cC(G,S_3^{(1)})$, with symmetry boundary $\cL(S_3^{(1)},1)$. The lattice model of this intrinsically beyond-Landau transition was studied numerically in~\cite{Warman:2026gfz}.

\subsection{Quantum Computing Gates}
\label{sec:ex9}

The final example script 
\begin{gapcode}
Read("examples/9_quantum_computing_D4N.g");
\end{gapcode}
computes the non-Clifford logical gates constructed in~\cite{Huang:2025ump,Warman:2025hov} using the quantum doubles of the dihedral groups. In this section only, the subscript follows these references and the script, so $D_{4N}$ is the group of symmetries of the regular $4N$-gon, of order $8N$, with $N=1,2,4$ (here, $D_4$ is the group of order $8$ that was called $D_8$ in Sec.~\ref{sec:ex2}). These can be used to topologically implement the gates
\be \label{eq:T1N}
T^{1/N}=\diag(1,e^{i\pi/(4N)})\,.
\ee
\paragraf{Hybrid Lattice Surgery.} \texttt{D4NSurgeryGate} runs the hybrid lattice surgery protocol of~\cite{Huang:2025ump}: a $\Z_{4N}$ patch prepared in the $\Z_{4N}$ stabilizer state 
$\ket{S}=\frac{1}{\sqrt{4N}}\sum_je^{i\pi j^2/(4N)}\ket{m^j}$, is teleported into a $D(D_{4N})$ patch by gauging the diagonal subgroup $\langle(m,r)\rangle\subset \Z_{4N}\times D_{4N}$. The topological orders are then split. $D_{4N}$ is then merged into a $\Z_2$ patch by gauging $\langle(r^{2N},m),(rs,1)\rangle$ and then $(s,1)$ to change boundary conditions: this uses the source code function \texttt{RoughMergeSplit}.

\paragraf{SPT-stacking.} A constant-depth construction of the gates \eqref{eq:T1N} was provided in \cite{Warman:2025hov} from an automorphism of $D(D_{4N})$ given by the non-trivial 2-cocycle class $\alpha\in H^2(D_{4N},U(1))$ with counterterms $\beta^{(i)}$ on the three boundaries of a triangular patch, resulting in the logical gate
\be
U_{\alpha,\beta}=\diag\Bigl(1,\frac{\alpha(g_1,g_2)\beta^{(3)}(g_1g_2)}{\beta^{(1)}(g_1)\beta^{(2)}(g_2)}\Bigr)\,,
\ee
which is computed in AlgebrasTQD by \texttt{SPTStackingGate}.
\texttt{D4NCocycle} returns $\alpha$ and the counterterm $\beta^{(3)}$ on $\langle r\rangle$, and we take $\beta^{(1)}=\beta^{(2)}=1$ on $\langle rs\rangle$ and $\langle s\rangle$. 
\texttt{D4NPhaseGate} realizes $\beta^{(3)}$ as an explicit circuit of phase and multiply-controlled $S$ gates on $\log_2(4N)$ qubits, and returns its diagonal and its level in the Clifford hierarchy. 

\newpage
The code
\begin{gapcode}
trivial := function(g) return 1; end;;

for N in [1,2,4] do
    Print("\n\nD",4*N," quantum computing gates\n\n");

    surgery := D4NSurgeryGate(N);;
    Print("Hybrid lattice surgery logical gate: diag",surgery.eigenphases,"\n\n");

    coc := D4NCocycle(N);;
    data := coc.data;;
    r := data.r;; s := data.s;;
    stacked := SPTStackingGate(rec(alpha:=coc.alpha,
        beta:=[trivial,trivial,coc.beta],
        K:=[Group(r*s),Group(s),Group(r)],g:=[r*s,s]));;
    phase := D4NPhaseGate(coc);;
    if stacked <> phase.U then
        ErrorNoReturn("The counterterm circuit must encode the SPT-stacking gate");
    fi;
    Print("SPT-stacking logical gate (constant depth): diag",DiagonalOfMat(stacked),"\n\n");
    Print("Boundary gate: diag",DiagonalOfMat(phase.Mbeta),"\n\n");
    Print("Clifford hierarchy level: ",phase.level,"\n");
od;
\end{gapcode}
prints 
\begin{gapout}
D4 quantum computing gates
Hybrid lattice surgery logical gate: diag[ 1, E(8) ]
SPT-stacking logical gate (constant depth): diag[ 1, E(8) ]
Boundary gate: diag[ 1, E(8), 1, -E(8)^3 ]
Clifford hierarchy level: 3


D8 quantum computing gates
Hybrid lattice surgery logical gate: diag[ 1, E(16) ]
SPT-stacking logical gate (constant depth): diag[ 1, E(16) ]
Boundary gate: diag[ 1, E(16), E(8), E(16)^3, 1, -E(16)^5, -E(8)^3, -E(16)^7 ]
Clifford hierarchy level: 4


D16 quantum computing gates
Hybrid lattice surgery logical gate: diag[ 1, E(32) ]
SPT-stacking logical gate (constant depth): diag[ 1, E(32) ]
Boundary gate: diag[ 1, E(32), E(16), E(32)^3, E(8), E(32)^5, E(16)^3, E(32)^7, 1, -E(32)^9, -E(16)^5, -E(32)^11, -E(8)^3, -E(32)^13, -E(16)^7, -E(32)^15 ]
Clifford hierarchy level: 5
\end{gapout}
This shows that the two constructions give the same logical gate: $T$ for $N=1$, $T^{1/2}$ for $N=2$ and $T^{1/4}$ for $N=4$, at levels $3$, $4$ and $5$ of the Clifford hierarchy.

\newpage
\section{Conclusion and Outlook}
\label{sec:conclusion}

AlgebrasTQD computes the condensable algebras of twisted quantum doubles $D^\omega(G)$ with their full algebra data and their anyon decomposition. The example scripts reproduce in minutes the results of the papers that motivated the code, and each code block of Sec.~\ref{sec:examples} can serve as a starting point for a new calculation. We close with some directions for future work.

\paragraf{Lattice Models and Tensor Networks.} Lattice Hamiltonians with group-theoretical symmetries, such as those of twin phases~\cite{Warman:2026gfz} and anyon chains~\cite{Aasen:2016dop,Aasen:2020jwb,Bhardwaj:2024kvy,Warman:2024lir,Inamura:2025cum}, can be built from the data that AlgebrasTQD computes. Combining this code with tensor-network methods~\cite{White:1992zz,Verstraete:2008cex,Schollwoeck:2010uqf,Orus:2013kga,Cirac:2020obd,Bultinck:2015bot,Lootens:2021tet,Lootens:2022avn,Lootens:2023wnl,Lootens:2024gfp} would allow simulations of phases and transitions beyond the reach of exact diagonalization. Such simulations could also check the features of twin transitions and verify that they indeed remain beyond Landau in every duality frame and distinguish them via their algebra of order parameters \eqref{eq:structure}, building on~\cite{Senthil:2003eed,Senthil:2023vqd,Zhang:2022wwn,Su:2024vrk,Chen:2025ivo,Chen:2026fai,Warman:2026gfz,Lootens:2021tet,Lootens:2022avn,Gai:2026hjk,Warman:2026gfz,Chavda:2026mpz}.

\paragraf{Condensable Algebras in Higher Dimensions.} For all-bosonic fusion 2-categories, describing symmetries of 2+1d theories related by generalized gauging to a finite group symmetry $G$ and anomaly $\pi\in H^4(G,U(1))$, the SymTFT is the $(3+1)$d Dijkgraaf-Witten theory $\cZ(2\Vec_G^\pi)$ \cite{Barkeshli:2014cna,Kong:2017hcw,Kong:2020cie,Bhardwaj:2024qiv,Bhardwaj:2025piv,Bhardwaj:2025jtf,Wen:2024qsg,Wen:2025thg}. Its ``minimal'' condensable algebras, in the terminology of~\cite{Bhardwaj:2024qiv}, are again labeled by group-theoretical data, analogous to $(H,F,\gamma,\eps)$. However, there are also ``non-minimal'' algebras, which involve $H$-graded (pre-)modular tensor categories and have no analogue in the $(2+1)$d SymTFTs of this paper. Extending AlgebrasTQD to these higher-dimensional SymTFTs, starting from the minimal algebras, is a natural next step.

\paragraf{Landau-Ginzburg Theories.} Field theories of Landau-Ginzburg type have recently been proposed for transitions with non-invertible symmetries in 1+1d~\cite{Chen:2025uno,Ravindran:2026jjn} and, coupled to gauge fields, for anyon proliferation in 2+1d~\cite{Cheng:2026qax,Cheng:2026hgu,Cheng:2026dxn,Schafer-Nameki:2026hgx,Gai:2026vgz,Eck:2026chq}. When the fields are associated with a condensable algebra, the algebra fixes more than the field content. For quantum doubles, its multiplication determines certain couplings, and twin algebras give distinct transitions~\cite{Gai:2026vgz}. The full algebra data computed by AlgebrasTQD could therefore provide the input of such theories for any finite group $G$ and $3$-cocycle $\omega$.

\paragraf{Quantum Error Correction.} A natural next step in the quantum information setting is a general investigation of all the logical gates of topological codes based on twisted quantum doubles, beyond the dihedral examples of Sec.~\ref{sec:ex9}. For twisted Abelian groups, such as $\Z_2^3$ with $\om_\III$, high-rate quantum low density parity check (qLDPC) generalizations of these codes have recently been constructed~\cite{Zhu:2026vec,Zhu:2026iep,Chang:2026xkz}. It would also be interesting to combine these logical gates with the exact recovery protocol of~\cite{Warman:2026flz} and with  decoders for non-Abelian codes~\cite{Sala:2024uom,Vadali:2026qon,Jing:2025zuu,delaFuente:2026slp,Sajith:2026vnw,Sala:2026lhy}.

\section*{Acknowledgments}
AW thanks Yuhan Gai and Sakura Sch\"afer-Nameki for the collaboration on twin algebras and twin phases~\cite{Gai:2026hjk,Warman:2026gfz} that motivated this code, Xiao-Gang Wen for suggesting GAP, and the developers of GAP and HAP~\cite{GAP4,HAP}. Anthropic's Claude Code and OpenAI's Codex were used for some re-formatting changes, code documentation descriptions and to suggest improvements. All AI-assisted output was directed, checked and verified by the author.

\paragraf{Funding information.} This work is supported by the UKRI Frontier Research Grant underwriting the ERC Advanced Grant ``Generalized Symmetries in Quantum Field Theory and Quantum Gravity'' of Sakura Sch\"afer-Nameki.
The author would like to thank the Isaac Newton Institute for Mathematical Sciences, Cambridge, for support and hospitality during the programme European tensor networks school, supported by EPSRC grant EP/Z000580/1, where work on this paper was completed.

\section*{Code Availability}
The code is available at \url{https://github.com/AlisonWarman/AlgebrasTQD}~\cite{AlgebrasTQD} under the Creative Commons Attribution 4.0 license. The repository also contains the modular-data file \texttt{source/0\_Modular\_Data\_} \texttt{Mignard-Schauenburg-Gruen-Morrison.g}, written by Mignard, Schauenburg, Gruen and Morrison~\cite{mignard2017moritaequivalencepointedfusion,gruen2021computing}, with small fixes and compatibility changes. Its original is part of the \TeX{} source of~\cite[arXiv:1808.05060]{gruen2021computing}, and its copyright remains with its authors.

\appendix

\section{Proofs of the Optimizations}
\label{app:proofs}

This appendix proves that the simplifications of Sec.~\ref{sec:solver} are valid, for every finite group and every normalized $3$-cocycle. We use the notation of Sec.~\ref{sec:theory}: ${}^gx=gxg^{-1}$, $\omega$ is a normalized 3-cocycle on $G$, and all cochains are normalized.

\begin{theorem}
\label{thm:gamma}
Let $F$ be a finite group, $\omega$ a normalized $3$-cocycle on $F$, and $S$ a generating set of $F$. For a normalized $2$-cochain $\gamma$ on $F$ the following are equivalent:
\begin{enumerate}
\item[(i)] $\gamma(f_2,f_3)\,\gamma(f_1,f_2f_3)=\gamma(f_1f_2,f_3)\,\gamma(f_1,f_2)\,\omega(f_1,f_2,f_3)$ for all $f_1,f_2,f_3\in F$,
\item[(ii)] the same equation holds for all $f_1,f_2\in F\setminus\{1\}$ and $f_3\in S\setminus\{1\}$.
\end{enumerate}
Moreover, a normalized $\gamma$ satisfying (i) exists if and only if $\omega$ is a coboundary, $\omega=d\gamma'$ for some $2$-cochain $\gamma'$.
\end{theorem}

\begin{proof}
(i) implies (ii) trivially. For the converse define
\be
    \beta(f_1,f_2,f_3)=\frac{\gamma(f_2,f_3)\,\gamma(f_1,f_2f_3)}{\gamma(f_1f_2,f_3)\,\gamma(f_1,f_2)\,\omega(f_1,f_2,f_3)}=\frac{d\gamma(f_1,f_2,f_3)}{\omega(f_1,f_2,f_3)}\,,
\ee
so that (i) is the statement $\beta=1$. The coboundary $d\gamma$ is a $3$-cocycle since $d\circ d=1$ on the standard cochain complex of $F$, and $\omega$ is a $3$-cocycle, so $\beta$ is a $3$-cocycle and satisfies \eqref{eq:cocycle}. It is normalized: since $\gamma$ is normalized, the remaining factors of $\gamma$ cancel in pairs and
\be
    d\gamma(1,f_2,f_3)=d\gamma(f_1,1,f_3)=d\gamma(f_1,f_2,1)=1\,,
\ee
Furthermore, $\omega(1,f_2,f_3)=\omega(f_1,1,f_3)=\omega(f_1,f_2,1)=1$ since $\omega$ is normalized. Then, assumption (ii) and the normalization properties imply 
\be
\beta(f_1,f_2,s)=1 \quad\text{for all } f_1,f_2\in F\quad\text{and}\quad s\in S\cup\{1\}\,.
\ee
Next, we define
\be
    T=\{t\in F\,:\,\beta(f_1,f_2,t)=1\text{ for all }f_1,f_2\in F\}\,,
\ee
which contains $1$ and $S$. If $t\in T$ and $s\in S$, the cocycle condition \eqref{eq:cocycle} of $\beta$ at $(f_1,f_2,t,s)$,
\be
    \beta(f_2,t,s)\,\beta(f_1,f_2t,s)\,\beta(f_1,f_2,t)=\beta(f_1f_2,t,s)\,\beta(f_1,f_2,ts)\,,
\ee
has every factor equal to $1$ except $\beta(f_1,f_2,ts)$, which is therefore $1$, and thus $ts\in T$. By induction on $l$, $T$ contains every product $s_1s_2\cdots s_l$ of elements of $S$. Since $F$ is finite, the inverse of $s\in S$ is the product $s^{\text{order}(s)-1}$, so every element of $F=\langle S\rangle$ is a product of elements of $S$. Hence $T=F$, i.e. $\beta=1$, which proves (i).

The last statement is standard, and we review its proof. If (i) holds then $\omega=d\gamma$ is a coboundary. Conversely let $\omega=d\gamma'$. Evaluating $d\gamma'$ at $(1,1,f)$ and at $(f,1,1)$ gives
\be
    \frac{\gamma'(1,f)}{\gamma'(1,1)}=\omega(1,1,f)=1\,,\qquad \frac{\gamma'(1,1)}{\gamma'(f,1)}=\omega(f,1,1)=1\,,
\ee
so $\gamma'(1,f)=\gamma'(f,1)=\gamma'(1,1)$ for all $f$. The constant $2$-cochain $\kappa=\gamma'(1,1)$ has $d\kappa=1$, so $\gamma=\gamma'/\kappa$ is normalized and $d\gamma=d\gamma'/d\kappa=\omega$.
\end{proof}

\noindent \texttt{GammaTrivialization} imposes (ii) as the linear congruences \eqref{eq:additive} in the unknowns $\gamma(f_1,f_2)$, $f_1,f_2\neq1$, with the terms that have an argument $1$ set to $1$, and solves them with \texttt{SolveModZ}, which returns a normalized $\gamma$ with $d\gamma=\omega|_F$ on all of $F^3$ when $\omega|_F$ is a coboundary, and \texttt{fail} otherwise.

\begin{lemma}[$\eps$ related by linear characters are equivalent]
\label{lem:gauge}
Let $F\triangleleft H$, $\gamma$ a normalized $2$-cochain with $d\gamma=\omega|_F$, and $\chi:F\to U(1)$ a homomorphism. If $\eps$ satisfies the three relations \eqref{eq:conditions}, so does $\eps^\chi(h,f)=\eps(h,f)\,\chi({}^hf)/\chi(f)$, with the same $\gamma$, and $A(H,F,\gamma,\eps^\chi)\cong A(H,F,\gamma,\eps)$. Moreover, $\eps^\chi=\eps$ on $F\times F$, and $f\mapsto\chi({}^hf)/\chi(f)$ is a linear character of $F$.
\end{lemma}

\begin{proof}
This is the case $g=1$ and $c=\chi$ of the isomorphism criterion \eqref{eq:isomorphism}~\cite{Davydov2009ModularIF,davydov2017lagrangian,Gai:2026hjk}. The last two statements follow from $\chi({}^{f_1}f_2)=\chi(f_2)$ and from $f\mapsto{}^hf$ being an automorphism of $F$.
\end{proof}

\begin{theorem}[$\eps$ from generators]
\label{thm:eps}
Let $F\triangleleft H\leq G$, $\gamma$ a normalized $2$-cochain with $d\gamma=\omega|_F$, and $S_H$ a generating set of $H$.
\begin{enumerate}
\item[(a)] A solution $\eps$ of \eqref{eq:conditions} has $\eps(1,f)=1$ for all $f\in F$, and for every $s\in H$ it satisfies \eqref{eq:projective}. If $k_s$ is one solution of \eqref{eq:projective}, then $\eps(s,\cdot)=k_s\chi_s$ for a homomorphism $\chi_s:F\to U(1)$.
\item[(b)] Let $\eps$ and $\eps'$ be functions on $H\times F$ that satisfy the first relation of \eqref{eq:conditions} for all $h_1\in H$, $h_2\in S_H$ and $f\in F$, with $\eps(1,\cdot)=\eps'(1,\cdot)$ and $\eps(s,\cdot)=\eps'(s,\cdot)$ for $s\in S_H$. Then $\eps=\eps'$.
\item[(c)] \texttt{epsf} returns exactly one representative of every equivalence class of the solutions $\eps$ to \eqref{eq:conditions} under Lemma~\ref{lem:gauge}.
\end{enumerate}
\end{theorem}

\begin{proof}
(a) The first relation at $h_1=h_2=1$ gives $\eps(1,f)=\eps(1,f)^2\,\omega(1,1|f)$, and $\omega(1,1|f)=\allowbreak\omega(1,f,1)/\bigl(\omega(1,1,f)\,\omega(f,1,1)\bigr)=\allowbreak1$ by normalization, so $\eps(1,f)=1$. The second relation at $h=s$, divided by $\gamma({}^sf_1,{}^sf_2)\,\Omega_s(f_1,f_2)$, is \eqref{eq:projective}, and this holds for every $s\in H$. If $k_s$ and $\eps(s,\cdot)$ both solve \eqref{eq:projective}, their ratio $\chi_s$ satisfies $\chi_s(f_1)\chi_s(f_2)=\chi_s(f_1f_2)$, so it is a homomorphism.

(b) Let $T=\{h\in H\,:\,\eps(h,\cdot)=\eps'(h,\cdot)\}$, which contains $1$ and $S_H$. If $h\in T$ and $s\in S_H$, then for all $f\in F$
\be
    \eps(hs,f)=\eps(h,{}^sf)\,\eps(s,f)\,\omega(h,s|f)=\eps'(h,{}^sf)\,\eps'(s,f)\,\omega(h,s|f)=\eps'(hs,f)\,,
\ee
so $hs\in T$. As in the proof of Theorem~\ref{thm:gamma}, every element of the finite group $H$ is a product of elements of $S_H$, so $T=H$.

(c) We recall the steps used by \texttt{epsf} to construct $\eps$ as a solution to \eqref{eq:conditions}. A candidate is fixed by $\eps$ on $F\times F$, given by the third relation, and by $\eps(s,\cdot)=k_s\chi_s$ for $s\in S_H\setminus F$, for a choice of linear characters $\chi_s$. If \eqref{eq:projective} has no solution for some $s$, there is no $\eps$ by (a). Otherwise, \texttt{epsf} extends each candidate to $H$ by \eqref{eq:propagate}, which reaches every element of $H$ by the closure argument of (b), and discards it when an element receives two different values, which means that the candidate solution was inconsistent. Three properties follow.
\begin{enumerate}
\item[(1)] No solution is missed. By (b) a solution $\eps$ is fixed by its values on $F$ and $S_H$, which are a candidate solution by (a) and the third relation. Then, \texttt{epsf} builds $\eps$ from this candidate.
\item[(2)] No false solution is returned, since \texttt{epsf} tests the first two relations on all of $H\times H\times F$ and $H\times F\times F$, and the third holds by construction.
\item[(3)] Isomorphic solutions related by Lemma~\ref{lem:gauge} are returned once, since \texttt{epsf} keeps one candidate for each orbit of the linear characters, and this action maps candidates to candidates and, by (b), solutions to solutions.\qedhere
\end{enumerate}
\end{proof}

\noindent For Abelian $H$, \texttt{epsfA} runs the same construction with a set of independent generators of $H$, and ${}^hf=f$ makes the equivalence of Lemma~\ref{lem:gauge} trivial, so every solution is returned. Since ${}^sf=f$, the factors of $\gamma$ cancel in \eqref{eq:projective}, which becomes $\beta_s=1/\Omega_s$, and when $\omega|_H$ is identically $1$ it has the solution $k_s=1$, so that the candidates are the linear characters themselves.

\begin{lemma}[Isomorphism test]
\label{lem:isotest}
Fix $g\in G$ with $H'={}^gH$ and $F'={}^gF$. If $c_0:F\to U(1)$ solves the first line of \eqref{eq:isomorphism}, then its solutions are all and only all $c=c_0\chi$ with $\chi:F\to U(1)$ a homomorphism and \eqref{eq:isomorphism} holds for some $c$ if and only if
\be
    \frac{\eps'({}^{g}h,{}^{g}f)\,\Omega_g(h,f)}{\eps(h,f)\,\Omega_g({}^{h}f,h)}\,\frac{c_0(f)}{c_0({}^{h}f)}=\frac{\chi({}^{h}f)}{\chi(f)}
\ee
for some linear character $\chi$ of $F$ and all $h\in H$, $f\in F$.
\end{lemma}

\begin{proof}
If $c$ and $c_0$ solve the first line of \eqref{eq:isomorphism}, which fixes $c(f_1f_2)/(c(f_1)\,c(f_2))$, then $\chi=c/c_0$ satisfies $\chi(f_1f_2)=\chi(f_1)\chi(f_2)$, and conversely $c_0\chi$ is a solution for every homomorphism $\chi$. Substituting $c=c_0\chi$ in the second line of \eqref{eq:isomorphism} gives the displayed condition.
\end{proof}

\noindent \texttt{AlgebraClassesG} runs over all $g\in G$ with $H'={}^gH$ and $F'={}^gF$, computes $c_0$ with \texttt{IsTrivial2coc}, whose \texttt{fail} means that no $c$ exists for this $g$, and tests the displayed condition for every linear character of $F$. By \eqref{eq:isomorphism} and Lemma~\ref{lem:isotest} this decides whether the two algebras are isomorphic.

\phantomsection
\bibliographystyle{ytphys}
\small
\bibliography{GenSym}

\end{document}